\documentclass[11pt]{article}
\pdfoutput=1
\usepackage[utf8]{inputenc}
\usepackage[T1]{fontenc}

\usepackage{formatting}

\usepackage{shortcuts}

\begin{document} 
\begin{titlepage}
\begin{center}
\phantom{ }
\vspace{3cm}

{\bf \Large{A Tale of Two Hungarians: Tridiagonalizing Random Matrices}}
\vskip 0.5cm
Vijay Balasubramanian${}^{\dagger}$, Javier M. Magan${}^{\ddagger}$, Qingyue Wu${}^{*}$
\vskip 0.05in
\small{${}^{\dagger}$ ${}^{\ddagger}$ ${}^{*}$ \textit{David Rittenhouse Laboratory, University of Pennsylvania}}
\vskip -.4cm
\small{\textit{ 209 S.33rd Street, Philadelphia, PA 19104, USA}}
\vskip -.10cm
\small{  ${}^{\dagger}$ \textit{Santa Fe Institute,}}
\vskip -.4cm
\small{\textit{ 1399 Hyde Park Road, Santa Fe, NM 87501, USA}}
\vskip -.10cm
\small{ ${}^{\dagger}$ ${}^{\ddagger}$ \textit{Theoretische Natuurkunde, Vrije Universiteit Brussel}}
\vskip -.4cm
\small{\textit{Pleinlaan 2,  B-1050 Brussels, Belgium}}
\vskip -.10cm
\small{  ${}^{\ddagger}$ \textit{Instituto Balseiro, Centro At\'omico Bariloche}}
\vskip -.4cm
\small{\textit{ 8400-S.C. de Bariloche, R\'io Negro, Argentina}}

\begin{abstract}
The Hungarian physicist Eugene Wigner introduced random matrix models in physics to describe the energy spectra of atomic nuclei. As such, the main goal of Random Matrix Theory (RMT) has been to derive the eigenvalue statistics of matrices drawn from a given distribution.   The Wigner approach gives powerful insights into the properties of complex, chaotic systems in thermal equilibrium.  Another Hungarian, Cornelius Lanczos, suggested a method of reducing the dynamics of any quantum system to a one-dimensional chain by tridiagonalizing the Hamiltonian relative to a given initial state.   In the resulting matrix, the diagonal and off-diagonal Lanczos coefficients control transition amplitudes between elements of a distinguished basis of states.  We connect these two approaches to the quantum mechanics of complex systems by deriving analytical formulae relating the potential defining a general RMT, or, equivalently, its density of states, to the Lanczos coefficients and their correlations.  In particular, we derive an integral relation between the average Lanczos coefficients and the density of states, and, for polynomial potentials, algebraic equations that determine the Lanczos coefficients from the potential.  We obtain these results for generic initial states in the thermodynamic limit.  As an application, we compute the time-dependent ``spread complexity'' in  Thermo-Field Double states and the spectral form factor for Gaussian and Non-Gaussian RMTs.
\end{abstract}
\end{center}

\small{\vspace{3 cm}\noindent ${}^{\dagger}$vijay@physics.upenn.edu\\
${}^{\ddagger}$magan@sas.upenn.edu\\
${}^{*}$aqwalnut@sas.upenn.edu
}

\end{titlepage}

\setcounter{tocdepth}{2}

{\parskip = .4\baselineskip \tableofcontents}
\newpage


\section{Introduction}

It is natural to adopt a statistical approach to the study of very complex systems by fixing initially given, coarse-grained, physical data, and averaging over the space of systems compatible with this data. Universal properties are then approachable in the average. The most famous example of this philosophy concerns the statistical origin of thermodynamics. But there are many other applications. For example, in quantum mechanics the complexity of a system might prevent us from knowing the actual Hamiltonian, impeding us from deriving the energy spectrum and the time dependence of wavefunctions. Again, we can hope that by averaging over spaces of Hamiltonians with given coarse-grained data, we might be able to access important information about the system.

This philosophy was pioneered by Wigner \cite{10.2307/1970079,10.2307/1969956,10.2307/1970008}, who wanted to understand aspects of the spectra of heavy nuclei. Wigner postulated that, given the complexity of nucleon interactions, certain aspects of the spectrum such as the spacings between the eigenvalues would be universal, and well described by averaging the Hamiltonian over an ensemble. Mathematically, since quantum mechanical Hamiltonians are just matrices (possibly of infinite dimension) satisfying certain constraints (self-adjoint and accommodating potential symmetries), this statistical approach leads to the consideration of ensembles of ``random matrices''. Although random matrices had been studied previously in the literature \cite{10.2307/2331939,10.1214/aoms/1177728846}, Wigner's contributions are usually considered seminal in the field of Random Matrix Theory (RMT) (see \cite{Fyodorov:2011} for a brief history). Since its conception, RMT has mostly concentrated on the derivation of spectral statistics. It is quite striking that such statistics can be computed exactly for many ensembles of random matrices and associated measures; see \cite{akemann2011oxford} for a recent account.

Notwithstanding the success of RMT in explaining many aspects of the spectra of complex systems, applications to real time dynamics are scarce in the literature.\footnote{See \cite{akemann2011oxford} for a review of some applications to scattering and the S-matrix.} One reason is that in many problems of interest, not only the eigenvalues but also the structure of eigenstates potentially plays a role, complicating the analysis and formulation of the problem. From a broader perspective, the problem of solving the Schr\"{o}dinger equation for complex systems in non-equilibrium scenarios is hard, even numerically. The reason is the exponential growth of the dimension of the Hilbert space with the number of degrees of freedom. It would be very useful if RMT could be applied to these types of problems.

In this regard,  the Lanczos or Recursion method \cite{viswanath2008recursion} provides a powerful procedure for solving the Schr\"{o}dinger equation and elucidating the time dependence of wavefunctions. Intuitively, this approach optimizes the amount of information needed to evolve an initial state $\vert \psi\rangle$, given the Hamiltonian $H$ of the system. Relatedly, as shown recently \cite{SpreadC}, this method uses a basis that minimizes the spread of the wavefunction over its basis elements. It does so by exploiting the actual states with which the initial state mixes through time evolution, i.e., the states $H^n\,\vert \psi\rangle$, that appear in the Taylor expansion of the solution to the Schr\"{o}dinger equation. As reviewed below, we can use these states to build an orthonormal basis, simultaneously generating the minimum subspace of the Hilbert space that is needed for the computation of the time dependent wavefunction. This distinguished set of states is called the Krylov basis.  At a technical level, the method rests on the computation of the Hamiltonian in the Krylov basis, a procedure that produces the tridiagonal ``Hessenberg form'' of the Hamiltonian. Known numerical algorithms use Householder reflections \cite{2020SciPy-NMeth,LAPACK-Hessenberg} to compute Hessenberg forms of matrices. The non-zero tridiagonal entries of this form of a Hamiltonian are called the Lanczos coefficients. The advantage of using this basis is that the Schr\"{o}dinger equation simplifies and describes a one-dimensional chain dynamics, making it easier to find time-dependent wavefunctions.

Here, we aim to build a bridge between Wigner's ideas concerning the approximation of Hamiltonians through Random Matrix Theory, and the Lanczos approach for analyzing Hamiltonian evolution. In one direction, Lanczos might enable the expansion of RMT to non-equilibrium (time-dependent) scenarios in quantum mechanics. Conversely, RMT might improve our understanding of the late time dynamics of quantum chaotic systems. Technically, instead of focusing on spectral statistics, we must determine the statistics of the Lanczos coefficients, namely the statistics of the Hessenberg, or tridiagonal, form of random matrices. 

The Lanczos method has indeed seen a recent burst of activity in the context of the study of chaos and quantum black holes. In particular, as proposed in \cite{Parker:2018yvk} and further developed in \cite{Barbon:2019wsy,Avdoshkin:2019trj,Dymarsky:2019elm,Magan:2020iac,Jian:2020qpp,Rabinovici:2020ryf,Dymarsky:2021bjq,Kar:2021nbm,Caputa:2021sib,Kim:2021okd,Caputa:2021ori,Patramanis:2021lkx,Trigueros:2021rwj,Rabinovici:2021qqt,Hornedal:2022pkc,Bhattacharjee:2022vlt,Adhikari:2022whf,Muck:2022xfc,Fan:2022mdw,Rabinovici:2022beu,Bhattacharya:2022gbz,Liu:2022god}, it has been used to provide a notion of operator complexity, dubbed Krylov complexity. In this context Ref. \cite{Kar:2021nbm,Rabinovici:2022beu} use RMT techniques  to analyze certain aspects of Krylov complexity. The Lanczos approach has also been used to compute out-of-time-ordered four-point functions and Lyapunov exponents \cite{Kobrin:2020xms}. A geometric and group theory approach to the Lanczos method has been developed in \cite{Caputa:2021sib}, connecting it to the field of generalized coherent states \cite{coherent1,coherent2,coherent3}. Finally, the method has been used to analyze the long time dynamics of quantum states in chaotic systems \cite{SpreadC}, where a new measure of state complexity dubbed spread complexity was put forward (see \cite{Caputa:2022eye,Caputa:2022yju} for recent applications to topological phases of matter, and \cite{Bhattacharjee:2022qjw} for an analysis of weak ergodicity breaking).  Concretely,  \cite{SpreadC} used the Lanczos method numerically to find transparent relations between the wave function at late times, spread complexity, the so-called spectral form factor in RMT \cite{Guhr:1997ve}, and the universal ensembles of RMT. 
In the present article, we will analytically relate the distribution of Lanczos coefficients to the density of states in general Random Matrix Theories and for generic initial states.   As a special case, we will re-derive the Lanczos coefficients, real-time dynamics, and complexity growth of RMTs in Thermo-Field Double states that were previously computed in \cite{SpreadC}. We will apply our results to more general non-equilibrium (time-dependent) scenarios in a companion article \cite{usfuture}.

Six sections follow. In  Sections II and III we review pertinent aspects of RMT and the Lanczos approach respectively. These sections are brief and independent. They do not contain new results and can be omitted if the reader is familiar with those techniques.
Section IV uses the existence of a large system size limit to derive an integral relation between the density of states of a RMT and the average Lanczos coefficients.  Section V uses saddlepoint methods to derive an analytical relation between the potential of an RMT and the one and two-point functions of the Lanczos coefficients, averaged over the ensemble.  For polynomial potentials, this leads to algebraic equations determining the Lanczos coefficients.
All the analytical findings are verified by numerical computation.  Section VI applies our results to Thermo-Field Double states, rederiving some of the findings of \cite{SpreadC} concerning the Spectral Form Factor and spread complexity. We close the article with a discussion in Section VII.

After this paper was submitted we became aware of Ref. \cite{doi:10.1063/1.533010,PhysRevD.47.1640}, see  \cite{1995ZPhyB..99..101W,Witte_1997}, which contain related methods and results in the context of the thermodynamic limit of the Lanczos method, dubbed there ``the plaquette expansion''. Since large random matrices are a special case of a thermodynamic limit, some of our results can be translated to those of \cite{doi:10.1063/1.533010,PhysRevD.47.1640}. In terms of the specific application to Random Matrix Theory, our approach clarifies some of the constructions, thus allowing new analyses such as computation of the statistics of correlations of Lanczos coefficients.

\section{Aspects of Random Matrix Theory}

This section explains the relevant details of Random Matrix Theory and will help to contextualize some of the later results.  If the reader is familiar with RMT, generalized beta ensembles, and the tridiagonal Gaussian ensembles,  this section can be skipped.  For a more detailed account of RMT see \cite{akemann2011oxford}.

A Random Matrix Theory is defined by specifying the probability of finding a particular instance of a matrix in a given ensemble. For example, the Gaussian Unitary Ensemble (GUE) is an ensemble of Hermitian $N\times N$ matrices $H_{ij}$ with  measure
\be 
\frac{1}{Z_{\textrm{GUE}}}e^{-\frac{N}{2} \textrm{Tr}(H^2)}\;,
\ee
where $Z_{\textrm{GUE}}=2^{N/2}\pi ^{N^2/2}$, the partition function, normalizes the probability distribution. Further Gaussian examples include the Gaussian Orthogonal Ensemble (GOE), defined as an ensemble of real symmetric $N\times N$ matrices $H$ with  measure
\be 
\frac{1}{Z_{\textrm{GOE}}}e^{-\frac{N}{4} \textrm{Tr}(H^2)}
\;,\ee
and the Gaussian Symplectic Ensemble (GSE), defined as an ensemble of $N\times N$ Hermitian quaternionic matrices with  measure
\be 
\frac{1}{Z_{\textrm{GSE}}}e^{-N \textrm{Tr}(H^2)}\;.
\ee
These ensembles are often denoted by their Dyson index: $ \beta = 1$ for GOE, $ \beta = 2$ for GUE, and $ \beta = 4$ for GSE. The Dyson index counts the number of real components per matrix element. In these conventions, the variance of each off-diagonal matrix entry for all ensembles is fixed to $\sigma^2=1/N$.

More generally, we can consider the same ensembles but modify the measure. Instead of using the Gaussian measure, we can write the distribution as an exponential of a generic potential $V(H)$, by replacing $\textrm{Tr}(H^2) \to V(H)$.  If the potential is invariant under $H\rightarrow U\,H\,U^{\dagger}$, with $U$ a unitary matrix, then it is natural to diagonalize the matrix. This is a change of variables  $H\rightarrow (U,\Lambda)$, where $U$ is the diagonalizing matrix and $\Lambda$ is a diagonal matrix of eigenvalues. The statistics of $U$ and $\Lambda$ factorize; so, for a Gaussian ensemble we get the joint probability distribution of  eigenvalues
\be \label{jointg}
p(\lambda_1,\cdots ,\lambda_n)=Z_{\beta,N}\, e^{-\frac{\beta \,N}{4}\sum_{k}\lambda^2_k}\,\prod_{i<j}\vert \lambda_i-\lambda_j\vert^\beta\;,
\ee
where $Z_{\beta,N}$ normalizes the distribution, and we have used the Dyson index to write a single formula applying simultaneously to the GUE, GOE, and GSE ensembles.
The non-trivial aspect of this change of variables is the computation of the Jacobian, which famously gives rise to the Vandermonde determinant $\Delta\equiv\prod_{i<j}\vert \lambda_i-\lambda_j\vert^\beta$. Extensions to non-Gaussian measures are straightforward since the Jacobian remains the same, and we have to simply modify the exponent which arises from the potential. For example, for a polynomial potential
\be 
V(H)=\sum\limits_{n}\,v_n\,\textrm{Tr}(H^n)\;,
\ee
with some real constants $v_n$, the joint probability distribution is just
\be \label{jointb}
p(\lambda_1,\cdots ,\lambda_n)=Z_{\beta,N}\, e^{-\frac{\beta\,N}{4}\,\sum\limits_{n}\,v_n\,\sum_{k}\lambda^n_k}\,\prod_{i<j}\vert \lambda_i-\lambda_j\vert^\beta\;.
\ee
This distribution serves as the starting point for analyzing the spectrum. For example, consider the statistics of the eigenvalue density. For a single matrix, this is defined as
\be 
\rho(E)=\frac{1}{N}\sum\limits_{i}\,\delta(E-\lambda_i)\;,
\ee
where $\lambda_i$ are the eigenvalues. The eigenvalue density is a random variable as it depends on the random eigenvalues $\lambda_i$. We seek the  correlation functions
\be 
\overline{\rho(E_1)\,\rho(E_1)\,\cdots \,\rho(E_n)}\;.
\ee
Equivalently we can compute the Laplace transform of the eigenvalue density, namely the partition function
\be 
Z_\beta =\int_{0}^{\infty}\,dE\,\rho(E)\,e^{-\beta\,E}\;,
\ee
and compute the associated correlation functions
\be 
\overline{Z_{\beta_1}\,Z_{\beta_2}\,\cdots \,Z_{\beta_n}}\;.
\ee
We assumed here that the theory is stable so that the spectrum is bounded from below, and then shifted the minimum energy to zero without loss of generality.  This problem has been completely solved for a zoo of matrix models (see \cite{akemann2011oxford,https://doi.org/10.48550/arxiv.1510.04430}).

Importantly, the joint probability distribution of eigenvalues~(\ref{jointg}) makes sense for arbitrary positive real values of $\beta$. This extends the GOE, GUE, and GSE universality classes to  ``generalized $\beta$-ensembles''. Remarkably, for the Gaussian case, it was shown in \cite{Dumitriu_2002} that such joint probability distributions of eigenvalues equivalently arise from certain ensembles of tridiagonal random matrices whose entries are distributed as
\be \label{indl}
H_\beta=\frac{1}{\sqrt{\beta\,N}} \begin{pmatrix}
N(0,2) & \chi_{(N-1)\beta} &  & & \\
\chi_{(N-1)\beta} & N(0,2) & \chi_{(N-2)\beta} & & \\
 & \ddots & \ddots & \ddots & \\
 & & \chi_{2\beta} & N(0,2) & \chi_{\beta}\\
  & &  &\chi_{\beta} & N(0,2)
\end{pmatrix}\;,
\ee
where $N(r,s)$ are independent Gaussian random variables with mean $r$ and variance $s$, and the $\chi_r$ are independent chi-distributed random variables, i.e. with probability density function
\be 
p_{\chi_r}(x)= \frac{1}{2^{r/2-1}\,\Gamma (r/2)} \,x^{r-1}\,e^{-x^2/2}\;.
\ee
This tridiagonal rewriting of Gaussian matrix models is very useful for the numerical generation of random matrices since it only involves $\mathcal{O}(N)$ random numbers, instead of $\mathcal{O}(N^2)$ numbers in the usual approach.

\section{Lanczos approach to unitary evolution}\label{LancSec}

In quantum mechanical systems the time evolution of a  state $\vert \psi (t)\rangle$ (a vector in a Hilbert space) is determined by the Schr\"{o}dinger equation
\be
i\partial_t\vert \psi (t)\rangle=H\vert \psi (t)\rangle \; ,
\label{eq:se}
\ee
where $H$ is the Hamiltonian operator. The solution is
$
\vert \psi (t)\rangle= e^{-iHt}\,\vert \psi (0)\rangle\,.
$
Taylor expanding, one obtains
\be\label{exp}
\vert \psi (t)\rangle=\sum^\infty_{n=0}\frac{(-it)^n}{n!}\,\vert \psi_n\rangle\;,
\ee
where 
\be 
\vert \psi_n\rangle\equiv H^n\,\vert \psi(0) \rangle \;.
\ee 
Knowledge of $\vert \psi_n\rangle$ is then equivalent to knowledge of the time evolution. Although we can expand $\vert \psi (t)\rangle$ in the set of vectors $ \vert \psi_n\rangle $, the latter are neither orthogonal to each other nor normalized. To remedy this we can apply the Gram–Schmidt procedure to the $\vert\psi_n\rangle $. This  generates an ordered, orthonormal basis
$\mathcal{K}=\set{\ket{K_n}: n=0,1,2,\cdots}$
that expands the subspace of the Hilbert space explored by time development of $|\psi(0) \rangle \equiv |K_0\rangle$. The basis $\mathcal{K}$, typically known as the Krylov basis, may not expand the full Hilbert space, depending on the dynamics and the choice of the initial state.

The Krylov basis ${\cal K}$ can be derived via the Lanczos algorithm \cite{viswanath2008recursion}. Starting from $\vert \psi_n\rangle = H^n \,|\psi(0)\rangle$ this algorithm generates an orthonormal basis
$\mathcal{K}=\set{\ket{K_n}: n=0,1,2,\cdots}$ as: 
\be
|A_{n+1}\rangle=(H-a_{n})|K_n\rangle-b_n|K_{n-1}\rangle,\quad |K_n\rangle=b^{-1}_n|A_n\rangle\;,
\label{eq:Lrecursion}
\ee
where the Lanczos coefficients  $a_n$ and $b_n$ read
\be
a_n=\langle K_n|H|K_n\rangle,\qquad b_n=\langle A_n|A_n\rangle^{1/2}\;.
\label{eq:anbndef}
\ee
This iterative process has initial conditions $b_0 \equiv 0$ and $|K_0\rangle=|\psi(0)\rangle$ being the initial state. Notice that the Lanczos algorithm (\ref{eq:Lrecursion}) implies
\be\label{Hact}
H|K_n\rangle=a_n|K_n\rangle+b_{n+1}|K_{n+1}\rangle+b_n|K_{n-1}\rangle \; .
\ee
The Hamiltonian then becomes a tri-diagonal matrix in the Krylov basis
\be\label{triH}
H= \begin{pmatrix}
a_0 & b_1 &  & & & \\
b_1 & a_1 & b_2 & & &\\
& b_2 & a_2 & b_3 & & \\
& & \ddots & \ddots & \ddots & \\
& & & b_{N-2} & a_{N-2} & b_{N-1}\\
&  & &  &b_{N-1} & a_{N-1}
\end{pmatrix}\;.
\ee
For finite-dimensional systems, this tridiagonal form of the Hamiltonian is known as the ``Hessenberg form'' of the matrix. For finite dimensional matrices there are numerically stable algorithms for computing it, see \cite{SpreadC} for details. There is also a more general method for computing the Lanczos coefficients, which remains valid for infinite dimensional systems. It starts from the ``survival amplitude'', i.e., the amplitude that the state at time $t$ is the same as the state at time zero, see \cite{SpreadC} for a detailed account and references.

This algorithm for matrix tridiagonalization was originally conceived by Lanczos \cite{Lanczos1950AnIM} to aid in the computation of eigenvalues and eigenvectors. Once the Hamiltonian is in tridiagonal form there are more effective methods for solving the eigenvalue/vector problem.\footnote{The original Lanczos algorithm suffers from an instability. The generated Krylov basis is less and less orthonormal at each step of the algorithm. This can be cured by orthogonalizing with all previously generated vectors, and not with the last two only \cite{articleOj}.} Nowadays there are better algorithms for this problem, but the tridiagonal form can be used directly to solve for the time evolution of the quantum system, without finding eigenvalues/vectors. Indeed, the wavefunction in the Krylov basis can be obtained by exponentiating the Hessenberg form and applying it to the initial state. The tridiagonal form of the Hamiltonian then implies that the Schr\"{o}dinger equation (\ref{eq:se}), when written in the Krylov basis, takes the form
\be\label{SchrodingerEq}
i\partial_t\,\psi_n(t)=a_n\,\psi_n(t)+b_{n+1}\,\psi_{n+1}(t) + b_n\,\psi_{n-1}(t) \;.
\ee
We conclude that in this basis, any time evolution becomes a one-dimensional motion with a  ``hopping'' Hamiltonian. 

Finally, given $\psi_n(t)$ it is natural to analyze the average position in the Krylov chain
\be 
C(t) = C_\mathcal{K}(t) = \sum_{n} n \,\vert \psi_n(t)\vert^2 =
\sum_{n} n \,p_n(t)\;,\label{eq:spread_comp_def}
\ee
and the effective dimension of the Hilbert space explored by the time evolution
\be 
C_{{\rm dim}}
= e^{H_\textrm{Shannon}}=e^{-\sum\limits_n p_n\log p_n}\;.
\label{ecom2}
\ee
In Ref.\cite{SpreadC} it was proven that these quantities, as computed in the Krylov basis, are a global minimum over different choices of basis. In this sense, they are sensible quantifications of complexity, understood as a measure of the spread of the wavefunction in the time evolved quantum state.

Examples in which the computation of the Lanczos coefficients $a_n$ and $b_n$ (the Hamiltonian in  tridiagonal form~(\ref{triH})) can be carried out analytically are sparse in the literature; see \cite{viswanath2008recursion} for old examples and \cite{Parker:2018yvk,Caputa:2021sib,SpreadC,Muck:2022xfc} for more recent ones. But there are extensive numerical applications of these techniques to different aspects of quantum mechanical theories. In particular, this approach has been used to study operator complexity in \cite{Parker:2018yvk,Barbon:2019wsy,Avdoshkin:2019trj,Dymarsky:2019elm,Magan:2020iac,Jian:2020qpp,Rabinovici:2020ryf,Dymarsky:2021bjq,Kar:2021nbm,Caputa:2021sib,Kim:2021okd,Caputa:2021ori,Patramanis:2021lkx,Trigueros:2021rwj,Rabinovici:2021qqt,Hornedal:2022pkc,Bhattacharjee:2022vlt,Adhikari:2022whf,Muck:2022xfc,Fan:2022mdw,Rabinovici:2022beu,Bhattacharya:2022gbz,Liu:2022god}, more  recently in connection with quantum chaos, random matrices and state complexity in \cite{SpreadC}, for the study of topological phases of matter \cite{Caputa:2022eye,Caputa:2022yju}, and for an analysis of weak ergodicity breaking \cite{Bhattacharjee:2022qjw}.

\section{Tridiagonalizing random matrices: a first approach}

Given the above comments about Random Matrix Theories and the Lanczos approach to non-equilibrium quantum mechanics, we now state our problem precisely. Given a Hermitian operator in a Hilbert space (a Hamiltonian) and a quantum state, we can apply the Lanczos method to obtain a tridiagonal matrix as in~(\ref{triH}). In RMT we have an ensemble of Hamiltonians, and thus the Lanczos algorithm will produce an ensemble of tridiagonalized random matrices. We thus seek to find the statistics of these tridiagonal matrices, namely the statistics of the Lanczos coefficients $a_n$ and $b_n$, given a particular RMT, defined by some potential $V(H)$ and/or its associated average density of states $\overline{\rho(E)}$.

Equivalently, given a RMT, we want to find the joint probability distribution for the Lanczos coefficients
\be 
p(a_0,\cdots ,a_{N-1},b_1,\cdots ,b_{N-1})\;,
\ee
together with averaged quantities such as
\be 
\overline{a_{m}\cdots a_{n} \,b_{r}\cdots b_{s}}\;.
\ee
In this section, we will identify the joint distribution of Lanczos coefficients for Gaussian theories, and derive an analytical formula for the average Lanczos coefficients in a generic RMT with an arbitrary potential.

\subsection{Exact examples: the Gaussian generalized \texorpdfstring{$\beta$}{beta}-ensembles}

It is instructive to start with examples that are both simple and exactly solvable. To this end, we start with the GOE. This is a Gaussian-distributed ensemble of orthogonal matrices:
\be 
\frac{1}{Z_{\textrm{GOE}}}e^{-\frac{N}{4} \textrm{Tr}(H^2)}
\label{eq:GOE2}
\;.\ee
As an initial state we choose
\be 
\vert\psi\rangle =(1,0,0,\cdots, 0)^{T}\;.
\label{eq:GOEinit}
\ee
These coefficients are given in the basis in which the Hamiltonian is a random matrix drawn from a GOE distribution.  This is a  generic choice of initial state since, given the orthogonal invariance of the GOE, we would obtain the same results for any initial state which is an orthogonal rotation of (\ref{eq:GOEinit}).

Since the initial state is always part of the Krylov basis, the Lanczos procedure is solved if we find a similarity transformation $O$ such that
\be 
O\,H\,O^{T}=\textrm{Tridiagonal}\,\,\,\,\,\,\,\,\,\, O\,\vert\psi\rangle=\vert\psi\rangle\;,
\label{eq:simO}
\ee
where the tridiagonal matrix has real entries and off-diagonal positive entries.  We require positive off-diagonal entries conventionally because the $b_n$ coefficients in the Lanczos procedure are chosen to be positive real numbers by a choice of phases in the Krylov basis elements (see eq.~\ref{eq:anbndef}). Since a similarity transformation takes orthonormal bases to orthonormal bases, the second relation in (\ref{eq:simO}) implies the initial state is part both of the initial basis and the new basis. The first relation in (\ref{eq:simO}) then implies that the new basis is the desired Krylov basis.

The second relation is fulfilled by any matrix of the form
\be\label{simU}
O= \begin{pmatrix}
1 & 0 \\
0 & M 
\end{pmatrix}\;,
\ee
where the zeroes represent $N-1$-dimensional vectors and $M$ is an $(N-1)\times (N-1)$ matrix. As shown by Dumitriu and Edelman \cite{Dumitriu_2002}, the tridiagonal form is then achieved as follows. Without loss of generality, suppose we draw an $N\times N$ matrix $H_N$ from the GOE ensemble. We can write it as
\be\label{Hsim}
H_{N}= \begin{pmatrix}
a_0 & x^{T} \\
x & H_{N-1} 
\end{pmatrix}\;,
\ee
where $x=(x_1,\cdots,x_{N-1})$ is a generic $(N-1)$ dimensional vector, $a_0$ is the first entry of the matrix and $H_{N-1}$ is an $(N-1)\times (N-1)$ matrix. We have called the first entry $a_0$ because, given the initial state, it coincides with the first Lanczos coefficient, namely
\be 
a_0=\langle\psi\vert\, H_{N}\,\vert\psi\rangle\;.
\ee
We will now construct the required similarity transformation in steps. We first choose $O$ of the form~(\ref{simU}) with any $(N-1)\times (N-1)$ orthogonal matrix $M$ such that
\be 
M\,x=\vert\vert x\vert\vert_2\,(1,0,0,\cdots, 0)^{T}\equiv \vert\vert x\vert\vert_2\,e_1^T\;,
\ee
where $\vert\vert x\vert\vert_2$ is the norm of $x$. This brings the matrix $H_{N}$ to the following form
\be 
O\,H_N\,O^{T}=\begin{pmatrix}
a_0 & \vert\vert x\vert\vert_2 \,e_1 \\
\vert\vert x\vert\vert_2 \,e_1^{T} & M\, \,H_{N-1}\,M^{T}
\end{pmatrix}\;.
\label{eq:GOEtransform1}
\ee

We can argue as follows that the statistics of $a_0$, $\vert\vert x\vert\vert_2$ and $M\, \,H_{N-1}\,M^{T}$ are uncorrelated given the statistics of the GOE in (\ref{eq:GOE2}).  Note first that the probability distribution in (\ref{eq:GOE2}) is proportional to $\exp\gr{-N/4 \, {\rm Tr}(H^2)}$ with a symmetric $H$ and hence is independently Gaussian in each entry of $H$ up to the symmetricity constraint.   Thus, since the entry $a_0$ is unchanged after the transformation (\ref{eq:GOEtransform1}), it continues to be  Gaussian distributed with the same mean and variance as the GOE, namely $2/N$ in our normalization since it belongs to the diagonal. The norm $\vert\vert x\vert\vert_2$ is then the square root of the sum of uncorrelated Gaussian random variables with zero mean and variance equal to the off-diagonal entries of the GOE, which is $1/N$. This gives rise to the chi-distribution defined earlier: $\chi_{N-1}/\sqrt{N}$. Finally, since $H_{N-1}$ is by definition a random matrix from the GOE ensemble of $(N-1)\times (N-1)$ matrices, and this ensemble is invariant under orthogonal transformations, $M\, \,H_{N-1}\,M^{T}$ is just a random matrix from the GOE.\footnote{Note that $H_{N-1}$ and/or $M\, \,H_{N-1}\,M^{T}$ belong to the GOE ensembles of $(N-1)\times (N-1)$ orthogonal matrices, but normalized such that the off-diagonal entries have variance equal to $1/N$, since they descend from the original $H_{N}$ which was normalized in such a way.} Repeating this procedure along the diagonal direction, and remembering that off-diagonal variances are fixed to $1/N$, we arrive at a tridiagonal matrix whose Lanczos coefficients have statistics given by
\be 
H_N=\begin{pmatrix}
a_0 & b_1 &  & &  \\
b_1 & a_1 & b_2 & & \\
 & \ddots & \ddots & \ddots & \\
 & & b_{N-2} & a_{N-2} & b_{N-1}\\
  & &  &b_{N-1} & a_{N-1}
\end{pmatrix}=\frac{1}{\sqrt{N}} \begin{pmatrix}
N(0,2) & \chi_{(N-1)} &  & & \\
\chi_{(N-1)} & N(0,2) & \chi_{(N-2)} & & \\
 & \ddots & \ddots & \ddots & \\
 & & \chi_{2} & N(0,2) & \chi_{1}\\
  & &  &\chi_{1} & N(0,2)
\end{pmatrix}\;.\nonumber
\ee
The same argument applies to the GUE and GSE and the Lanczos algorithm will give rise to the following Lanczos coefficients
\be 
H_N=\begin{pmatrix}
a_0 & b_1 &  & &  \\
b_1 & a_1 & b_2 & & \\
 & \ddots & \ddots & \ddots & \\
 & & b_{N-2} & a_{N-2} & b_{N-1}\\
  & &  &b_{N-1} & a_{N-1}
\end{pmatrix}=\frac{1}{\sqrt{\beta\,N}} \begin{pmatrix}
N(0,2) & \chi_{(N-1)\beta} &  & & \\
\chi_{(N-1)\beta} & N(0,2) & \chi_{(N-2)\beta} & & \\
 & \ddots & \ddots & \ddots & \\
 & & \chi_{2\beta} & N(0,2) & \chi_{\beta}\\
  & &  &\chi_{\beta} & N(0,2)
\end{pmatrix}\;.\nonumber
\ee
We conclude that, for the Gaussian Orthogonal, Unitary and Symplectic ensembles, and generic initial states, the $a_n$ Lanczos coefficients are independent random variables with zero mean and variance equal to the one of the RMT. The $b_n$ Lanczos coefficients are also independently distributed random variables with distribution
\be 
p(b_n)=2\, \left( \frac{\beta\,N}{2}\right) ^{(N-n)\beta/2}\,\frac{1}{\Gamma ((N-n)\beta/2)} \,b_n^{(N-n)\beta-1}\,e^{-\beta\,N\,b_n^2/2}\;.
\ee
The average of the $b_n$ Lanczos coefficients is then
\be 
\overline{b_n}=\sqrt{\frac{2}{\beta\,N}}\,\frac{\Gamma ((1+(N-n)\beta)/2)}{\Gamma ((N-n)\beta/2)}\;,
\ee
while the variance reads
\be 
\sigma^2=\overline{(b_n-\overline{b_n})^2}=\frac{(N-n)}{ N}-\overline{b_n}^2\;.
\ee
We verify this numerically  in  Fig~(\ref{fig:ex_ab}).

Finally, notice that the generalized $\beta$-ensembles for any $\beta>0$ define  Hamiltonians with sensible Lanczos coefficients given by their tridiagonal version. 

\subsection{The average Lanczos coefficients for generic RMT: a physicist's argument}\label{rho_to_bx}

Above we used the work of Dumitriu and Edelman \cite{Dumitriu_2002}, to explain the tridiagonal form, and hence the Lanzos coefficients, of Gaussian random matrix ensembles.  We now generalize these results to generic Random Matrix Theories with generic densities of states $\rho(E)$, and compute the statistics of the associated tridiagonal matrices, namely the statistics of the Lanczos coefficients.


Notice first that, in numerical computations (see \cite{SpreadC} and below), and in the exact analytical solution above for the Gaussian case, the Lanczos coefficients $a_n, b_n$, when expressed as a function of $x=n/N$, have a continuous large-$N$ limit $a(x),b(x)$, where $N$ is the dimension of the Hilbert space. We will show that an approximate analytical formula relating the average Lanczos coefficients to the density of states can be derived simply by assuming the existence of a continuous large-N limit.


Let $H$ be our Hamiltonian, and $a_n, b_n$ be its Lanczos coefficients when starting from some initial state. In the Krylov basis, the dynamics is that of a 1-D chain~(\ref{SchrodingerEq}). We cut this 1-D Krylov chain into many (say $S=\sqrt{N}$) shorter ``segments'' (of length, say $L=\sqrt{N}$). This modification is accomplished by setting $b_n$ to zero at the boundaries of each segment and then setting the $a_n$ and $b_n$ within each segment to be their average across the segment. We claim that for large $N$, this process approximately preserves the density of states. The intuition is that for large $N$, each of these segments is long, so the first step of cutting the 1-D chain into segments only introduces a small effect from the edges of the segments. The second step of setting the Lanczos coefficients to their average also only introduces a small effect because we are assuming that the $a_n$'s and $b_n$'s, when expressed as a function of $x=n/N$, have a continuous large-N limit.  \footnote{We are using this trick only to approximate the total density of states. At the end of the day, the Lanczos coefficients we will find do not vanish anywhere, except at the end of the Krylov chain.}  To have such a limit, the Lanczos coefficients must be sufficiently slowly changing in any segment of length $L$ such that $L/N \to 0$ as $N\to \infty$.  Below, we will numerically confirm the validity of this approximation and also provide a more precise argument making use of the moments of the Hamiltonian.

We thus proceed to compute the Lanczos coefficients $a(x)$ and $b(x)$, as a function of $x=n/N$ for $0\leq x\leq 1$ as $N \to \infty$. These are the natural variables in the large-$N$ limit, and we seek to compute their averages $\overline{a(x)},\overline{b(x)}$ over the matrix ensemble.  For notational simplicity,  we will omit the overline denoting the average so that the functions $a(x)$, $b(x)$ and $\rho(E)$ denote the average Lanczos coefficients and density of states respectively.

In the block approximation of the Hamiltonian, the density of states is simply the sum of the densities of states of each of the $S$ segments. Each segment, of size $L$, has approximately constant $a$ and $b$, is therefore Toeplitz.  A standard formula from linear algebra then tells us the block has eigenvalues $
E_k=2\,b\,\cos (k\pi/(L+1)) +a $, with $k=1\cdots L$.  Recall also that we can take $a$ to be real and $b$ to be positive given the Lanczos algorithm (\ref{eq:anbndef}). The density of states in such a segment is then
\be 
\rho_{a,b}(E)=\frac{1/L}{|dE_k/dk|}= \frac{H(4\,b^2-(E-a)^2)}{\pi\,\sqrt{4\,b^2-(E-a)^2}}\;,
\ee
where $H(x)$ is the Heaviside step function setting the density of states to zero outside of its domain.  Here we divided by $L$ to normalize the density of states to integrate to 1 over its domain, as conventional in Random Matrix Theory. In the large-$N$ limit, both $S=\sqrt{N}$ and $L=\sqrt{N}$ are large as well. Noticing that $N=SL$, the total (normalized) density of states is then approximated by an integral
\be \label{intdl}
\rho(E) = \frac{1}{N}\sum_{n=1}^{S} \frac{L \,  H(4\,b(nL)^2-(E-a(nL))^2)}{\pi\,\sqrt{4\,b(nL)^2-(E-a(nL))^2}} = \int_0^1 dx\, \frac{H(4\,b(x)^2-(E-a(x))^2)}{\pi\, \sqrt{4\,b(x)^2-(E-a(x))^2}}\;.
\ee
This formula explicitly relates the average Lanczos coefficients to the density of states in the large-N limit when the block approximation of the Hamiltonian is valid, and $x$ is not too close to the edges $x=0,1$, where by close we mean $x\sim\mathcal{O}(1/N)$. Eq.~\ref{intdl} is one of the main results of this article. 

A somewhat more precise argument for this formula makes use of the moments of the Hamiltonian.  Notice that because the Hamiltonian is tridiagonal, $[H^n]_{ij}$, the $i,j$th entry of $H^n$, is an $n$-th order polynomial of some $a_j,b_j$ with $|j-i|<n$. Let $k$ be an index within $n$ of $i,j$ and let $n$ scale sub-linearly in $N$ in the large $N$ limit so that the difference in $x$ satisfies $|i-k|/N<n/N\to 0$,  Then assuming as above that $a(x)$ and $b(x)$ have continuous large $N$ limits, we can approximate all instances $a_i,b_i$ with $a_k,b_k$ instead, and so\footnote{This approximation applies far from the first rows and columns of $H$.  Namely, it does not apply when $i$ is $O(1)$ in the large $N$ limit (equivalently when $x\sim\mathcal{O}(1/N)$).}
\be
[H^n]_{ij}\,\approx\,  [T(a_k,b_k)^n]_{ij}\;,\label{eq:local_approx}
\ee
where $T(a,b)$ is an infinite tridiagonal matrix\footnote{We are not saying that $H^n$, a finite dimensional matrix, is approximately equal to an infinite dimensional matrix. We are just asserting the approximate equality of certain elements in both matrices.} with constant diagonal $a$ and off-diagonal $b$
\be\label{Triav}
    T(a,b)=\begin{pmatrix}\ddots&\ddots&&&\\
    \ddots&a&b&&\\
    &b&a&b&\\
    &&b&a&\ddots\\
    &&&\ddots&\ddots
    \end{pmatrix}.
\ee
Therefore, the trace of the moment of $H$ is 
\be
    \tr{H^n}=\sum_i [H^n]_{ii}\approx \sum_i [T(a_i,b_i)^n]_{ii}\;.\label{eq:trace_scaleshiftsolve}
\ee
Notice that the {\it matrices} on the right hand sides of \eqref{eq:local_approx} and \eqref{eq:trace_scaleshiftsolve} depend on $i$; for every index $i$ on the left hand side, we make this approximation with a different matrix.

Now we make some observations about the matrix $T(a,b)$. First, by counting Dyck paths one can verify that when $a=0$ and $b=1$,
\be
    [T(0,1)^n]_{ij}=\binom{n}{(n+ (i-j))/2},\label{eq:trid_to_binomial}
\ee
where the binomial is taken to be zero when $(n+(i-j))/2$ is not an integer. We are also going to use the integral identities 
\bea
 \binom{n}{n/2}  &=&  \int_{-2}^2 dx\,\frac{x^n}{\pi \sqrt{4-x^2}} = [T(0,1)^n]_{ii}\;,   \nonumber \\ 
\binom{n}{(n+1)/2} &=& \int_{-2}^2 dx\,
   \frac{x^{n+1}}{2\pi \sqrt{4-x^2}}
= [T(0,1)^n]_{i(i+1)} 
   \;.
   \label{eq:trid_integral}
\eea
We will only need the first expression
in this section, but include the $i-j=1$ case in the second expression for use in later sections. 

Second, we note that $T(a,b)$ is related to $T(0,1)$ by a scaling and a shift
\be 
T(a,b)=b\, T(0,1)+a \, .
\ee
Using this relation and the integral expression of $[T(0,1)^n]_{ii}$, we can expand 
\begin{equation}
[T(a,b)^n]_{ii} = \int_{-2}^2 dx \, \sum_k \binom{n}{k} b^k a^{n-k}  \frac{x^k}{\pi \sqrt{4 - x^2}}\;.
\end{equation}
Adding  powers of $x^k$ and substituting $E = bx + a$, and similarly treating $[T(a,b)^n]_{(i+1)i}$, we find
\bea
    [T(a_i,b_i)^n]_{ii}&=& \int_{a_i-2b_i}^{a_i+2b_i} dE\,\frac{E^n}{\pi \sqrt{4b_i^2-(E-a_i)^2}}\;,\\
    ~[T(a_i,b_i)^n]_{i,i+1}&=& \int_{a_i-2b_i}^{a_i+2b_i} dE\,\frac{E^n(E-a_i)}{b_i\pi \sqrt{4b_i^2-(E-a_i)^2}}\label{eq:trid_approx} \, .
\eea
Thus equation \eqref{eq:trace_scaleshiftsolve} becomes
\bea\label{secapp}
    \int dE\,E^n\rho(E)&\approx& \frac{1}{N}\sum_i \int_{a_i-2b_i}^{a_i+2b_i} dE\, \frac{E^n}{\pi \sqrt{4b_i^2-(E-a_i)^2}}\\
    &\approx& \int_0^1 dx\,\int_{a(x)-2b(x)}^{a(x)+2b(x)} dE\, \frac{E^n}{\pi \sqrt{4b(x)^2-(E-a(x))^2}}\;,\;
\eea
where we have taken the large $N$ limit and used the same definitions of $a(x),b(x)$ as above. Note that we always use the convention that $\rho(E)$ is normalized so that its integral is one. Given~(\ref{secapp}), and the fact that the polynomials form a complete basis of functions, we can infer that
\bea
    \rho(E)\approx \int_0^1 dx\,\frac{H(4\,b(x)^2-(E-a(x))^2)}{\pi \sqrt{4\,b(x)^2-(E-a(x))^2}}\;,\label{eq:intg_eqn_abrho}
\eea
arriving again at the relation (\ref{intdl}) between the density of states and the Lanczos coefficients. In the next two sections, we explain how to solve this equation and verify the results for specific examples numerically.

\subsubsection{Solving the integral equation}
\label{solve_intdl}

We now explain a strategy for solving the integral equation (\ref{eq:intg_eqn_abrho}),  thus deriving the Lanczos coefficients from the density of states, under the assumption that the interval of support in $E$ of the integrand of \eqref{eq:intg_eqn_abrho}  shrinks monotonically as $x$ increases. In fact, it is known that  the outer envelope of the Lanczos coefficients does contract monotonically  \cite{doi:10.1063/1.533010,PhysRevD.47.1640}. Concretely, the Heaviside function dictates that a given value of $x$ in the integrand of \eqref{eq:intg_eqn_abrho} contributes to the density of states for energies in the range $a(x) - 2 b(x) \leq E \leq a(x) + 2 b(x)$.  We will assume that if $x_1 > x_2$ then $a(x_1) +2 b(x_1) < a(x_2) +2 b(x_2)$ and $a(x_1) - 2 b(x_1) > a(x_2) - 2 b(x_2)$, or, infinitesimally, 
$a'(x) + 2b'(x) < 0$ and $a'(x) - 2 b'(x) > 0$.  This implies that $ 2b'(x) \leq a'(x) < -2 b'(x)$.  Solutions to this constraint require $b'(x) < 0$ and $|a'(x)| < -2 b'(x)$.
Solutions to the integral equation \eqref{eq:intg_eqn_abrho} that correspond to the Lanczos coefficients at large $N$ obey this monotonicity assumption for all examples that we have checked that also satisfy the assumption of a continuous large $N$ limit.\footnote{It would useful to give a  proof establishing physical conditions under which the monotonicity condition holds. Note also that the derivation of \eqref{eq:intg_eqn_abrho} did not require this assumption.}

We start with theories for which $a(x)=0$.  By inspection, this scenario (\ref{eq:intg_eqn_abrho}) corresponds to Hamiltonians whose density of states $\rho(E)$ is even in $E$.  In fact, the converse is also true -- if the density of states is even, then $a(x) = 0$.  This follows because the Heaviside function in \eqref{eq:intg_eqn_abrho} determines a non-even range of $E$  with a a non-vanishing density of states if $a(x) \neq 0$.\footnote{More generally  we only need the density of states to be symmetric around a certain energy $E_0$. If this is the case we just shift so that $E_0\rightarrow 0$ to arrive at an even density of states.} In this case, the integral equation can be solved by deconvolving via Laplace transforms. The monotonicity condition tells us that $b'(x)\leq 0$, so we may define $\gamma(b(x))\equiv b'(x)$, where $\gamma(b)$ can be thought as the density of values of $b$. Then a change of variables gives
\be 
\rho(E) = \int_{b_0}^{0} \, db\,\gamma(b)\, \frac{H(4\,b^2-E^2)}{\pi\,\sqrt{4\,b^2-E^2}}\;,
\ee
where, using the monotonicity assumption, $b$ goes from some initial $b_0$ at $x=0$ to $b=0$ at $x=1$, where the Lanczos algorithm must halt since we have reached the dimension of the Hilbert space, thus implying that $b(1) = 0$. Further substituting $b=b_0\,e^{-z}$ and $E=E_0\,e^{-\epsilon}$ we arrive at
\eqm{
\rho(E_0\,e^{-\epsilon}) &= \int_{0}^\infty \,dz\,b_0\,e^{-z}\,\gamma(b_0\,e^{-z})\,\frac{H(4-E_0^2\,b_0^{-2}\,e^{2(z-\epsilon)})}{\pi\, b_0\,e^{-z}\,\sqrt{4-E_0^2\,b_0^{-2}\,e^{2(z-\epsilon)}}}\\
&=\int_{0}^\infty \,dz\,b_0\,\gamma(b_0\,e^{-z})\,\frac{H(4-E_0^2\,b_0^{-2}\,e^{2(z-\epsilon)})}{\pi \,b_0\,\sqrt{4-E_0^2\,b_0^{-2}\,e^{2(z-\epsilon)}}}\;.}
This equation is of the form
\bea 
f(\epsilon)&=\int_{0}^\infty \,dz\,g(z)\,h(\epsilon-z)\;.
\eea
Since $f$ and $h$ are known, we can deconvolve, either numerically or analytically, using the Laplace transform to solve for $g(z)=b_0\,\gamma(b_0\,e^{-z})$. Then solving the differential equation $b'(x)=\gamma(b(x))$ gives us $b(x)$.

Next we consider the general case with $a(x)\neq 0$. Recall again that the support of $E$ in the integrand of \eqref{eq:intg_eqn_abrho} lies in the interval from $E_\text{left}(x)=a(x)-2b(x)$ to $E_\text{right}(x)=a(x)+2b(x)$, and that we have assumed that this interval shrinks as $x$ increases. To use this assumption we first integrate to compute the cumulative density of states
\eqref{eq:intg_eqn_abrho} to get
\bea
    P(E)=\int_{E_\text{min}}^EdE\,\rho(E)&\approx& \int_0^1 dx\int_{E_\text{left}(x)}^EdE\,\frac{H(4\,b(x)^2-(E-a(x))^2)}{\pi \sqrt{4\,b(x)^2-(E-a(x))^2}}\nonumber\\
    &=&\int_0^1 dx\,P_c\gr{4\frac{E-E_\text{left}\gr{x}}{E_\text{right}\gr{x}-E_\text{left}\gr{x}}-2 }\;,
    \label{eq:cumulativeint}
\eea
where $P_c(z)=\int_{-2}^z \frac{1}{\pi \sqrt{4-z'^2}}dz'$ and $P(E)=\int_{E_\text{min}}^EdE\,\rho(E)$ are  cumulative distributions.  The lower limit of the integral over $E$ in the first line is $E_\text{left}(x)$ because a given $x$ only contributes to the density of states for $E > E_\text{left}(x)$.  To arrive at the second line we used the substitution $E = b z' - a$.

Now suppose that we consider an $E$ that is equal to $E_\text{left}(X)$ for some $X$.  Then our monotonicity assumption says that the density of states for any $E < E_\text{left}(X) $ only takes contributions from $x \leq X$.  Thus the cumulative density of states from $E = E_\text{min}$ up to $E = E_\text{left}(X)$ only takes contributions from $x \leq X$.
Thus, we can 
we can reduce the limits of integration in (\ref{eq:cumulativeint}) to
\bea\label{keynum}
    P(E_\text{left}(X))
    &=&\int_0^X dx\,P_c\gr{4\frac{E_\text{left}(X)-E_\text{left}\gr{x}}{E_\text{right}\gr{x}-E_\text{left}\gr{x}}-2 }\;.
\eea
We may carry out the same steps but integrate from the other end to obtain an equation for $E_\text{right}(X)$ as well
\bea\label{keynum2}
    1-P(E_\text{right}(X))
    &=&\int_0^X dx\,\sq{1-P_c\gr{4\frac{E_\text{right}(X)-E_\text{left}\gr{x}}{E_\text{right}\gr{x}-E_\text{left}\gr{x}}-2 }}\;.
\eea
Since these equations only ``look into the past'' of $X$, we can iteratively solve this system of equations for $E_\text{left}(x)$ and $E_\text{right}(x)$  as in algorithm \ref{alg:scaleshiftsolve}, given the functions $P(E)$ (which can be determined from the density of states) and $P_c(z)$ (which can be determined by computation of the defining integral).  Then the Lanczos coefficients $a(x)$ and $b(x)$ are simply given by inverting the definitions of $E_\text{left,right}(x)$: $a(x) = (E_\text{right}(x) + E_\text{left}(x))/2 $ and $b(x) = (E_\text{right}(x) - E_\text{left}(x))/4 $.

Here, we assumed that the density of states is supported over a finite interval $[E_\text{min},E_\text{max}]$; if not, we can cut off the tail of the density of states as an approximation. The numerical algorithm works by discretizing $x$ into $M$ small intervals of size $1/M$, and assuming that $E_\text{left},E_\text{right}$ are constant over those intervals, so that the integrals in (\ref{keynum}, \ref{keynum2}) become discrete sums. Evaluating the resulting equations at the points $m/M+\epsilon$ gives us (\ref{eq:scaleshiftsolve_1}, \ref{eq:scaleshiftsolve2}), which can be used to solve for the values of $E_\text{left},E_\text{right}$ at $m/M$ from the values at $i/M$ for $i<m$.

\begin{algorithm}[H]
\caption{Approximating solutions to the integral equation}
\label{alg:scaleshiftsolve}
\begin{algorithmic}[1]
\State $E_\text{right}(0) \gets E_\text{max}$
\State $E_\text{left}(0) \gets E_\text{min}$
\For{$m\in 1:M$}
    \State Set $E_\text{left}\gr{\frac{m}{M}}$ be the lowest solution $E>E_\text{left}\gr{\frac{m-1}{M}}$ of 
    \bea
        P(E)=\frac{1}{M}\sum_{i=0}^{m-1} P_c\gr{4\frac{E-E_\text{left}\gr{\frac{i}{M}}}{E_\text{right}\gr{\frac{i}{M}}-E_\text{left}\gr{\frac{i}{M}}}-2}
        \label{eq:scaleshiftsolve_1}
    \eea
    \State Set $E_\text{right}\gr{\frac{m}{M}}$ be the highest solution $E<E_\text{right}\gr{\frac{m-1}{M}}$ of
    \bea
        1-P(E)=\frac{1}{M}\sum_{i=0}^{m-1} \sq{1-P_c\gr{4\frac{E-E_\text{left}\gr{\frac{i}{M}}}{E_\text{right}\gr{\frac{i}{M}}-E_\text{left}\gr{\frac{i}{M}}}-2}}
        \label{eq:scaleshiftsolve2}
    \eea
\EndFor
\State $a\gets \frac{E_\text{left}+E_\text{right}}{2}$
\State $b\gets \frac{E_\text{right}-E_\text{left}}{4}$
\end{algorithmic}

\end{algorithm}

To solve Eqs.~(\ref{eq:scaleshiftsolve_1} and \ref{eq:scaleshiftsolve2}) 
we used a bisection method to get a good initial estimate before using Newton's method to find an accurate solution. Note that the strict inequality for $E$ is necessary because $E=E_\text{left}\gr{\frac{n-1}{N}}$ and $E=E_\text{right}\gr{\frac{n-1}{N}}$ are themselves solutions to the equation. We will describe a few steps of the procedure to make it more intuitive. At $m=0$, solving $P(E)=0$ and $1-P(E)=0$ gives $E_\text{left}(0)=E_\text{min}$ and $E_\text{right}(0)=E_\text{max}$, respectively. The interval of support of the zeroth step is the same as the interval of support of the density of states. For $m=1$, we solve for the $E_\text{left}(1/M)$ such that $P(E_\text{left}(1/M))=\frac{1}{M}P_c\gr{4\frac{E_\text{left}(1/M)-E_\text{min}}{E_\text{max}-E_\text{min}}-2}$, and similarly for the $E_\text{right}$. The right-hand side is the contribution to the cumulative density of states of the first interval of the integral from $x=0$ to $x=1/M$. We then set the bounds of the next contribution to be the $E_\text{left},E_\text{right}$ where the first contribution is equal to the cumulative density of states.  This is where the first contribution starts undershooting the cumulative density of states and another contribution is needed to make up for it.

\subsubsection{Examples and numerical verification}
\label{sec:examples1}

First consider the Gaussian Unitary Ensemble, defined by a potential 
\begin{equation}
V_g(E)\equiv E^2 \;.
\end{equation}
The exact tridiagonalization of this theory was reviewed before, but here we arrive at the same results using the analytical relation between the density of states and the Lanczos coefficients (\ref{eq:intg_eqn_abrho}). In particular, for the GUE the density of states is given by the Wigner semicircle law
\be
\rho(E)=\frac{\sqrt{4-E^2}}{2\pi}\;.
\ee
Using this as the input to the integral equation we find the analytical solutions
\be 
a(x)=0\;,\,\,\,\,\,\,\,\,\,\,\,\,\, b(x)^2=(1-x)\;.
\ee
The leftmost plot in Fig.~(\ref{fig:ex_ab}) compares this result to numerically computed Lanczos coefficients for individual draws from the GUE ensemble.
The noisy fluctuations in this plot depict the actual Lanczos coefficients of specific instances of matrices in the corresponding ensemble, a computation that can be done using conventional stable routines to find the Hessenberg form of a matrix (see \cite{SpreadC} for details).

\begin{figure}[t]
\centering
\includegraphics[width=0.98\linewidth]{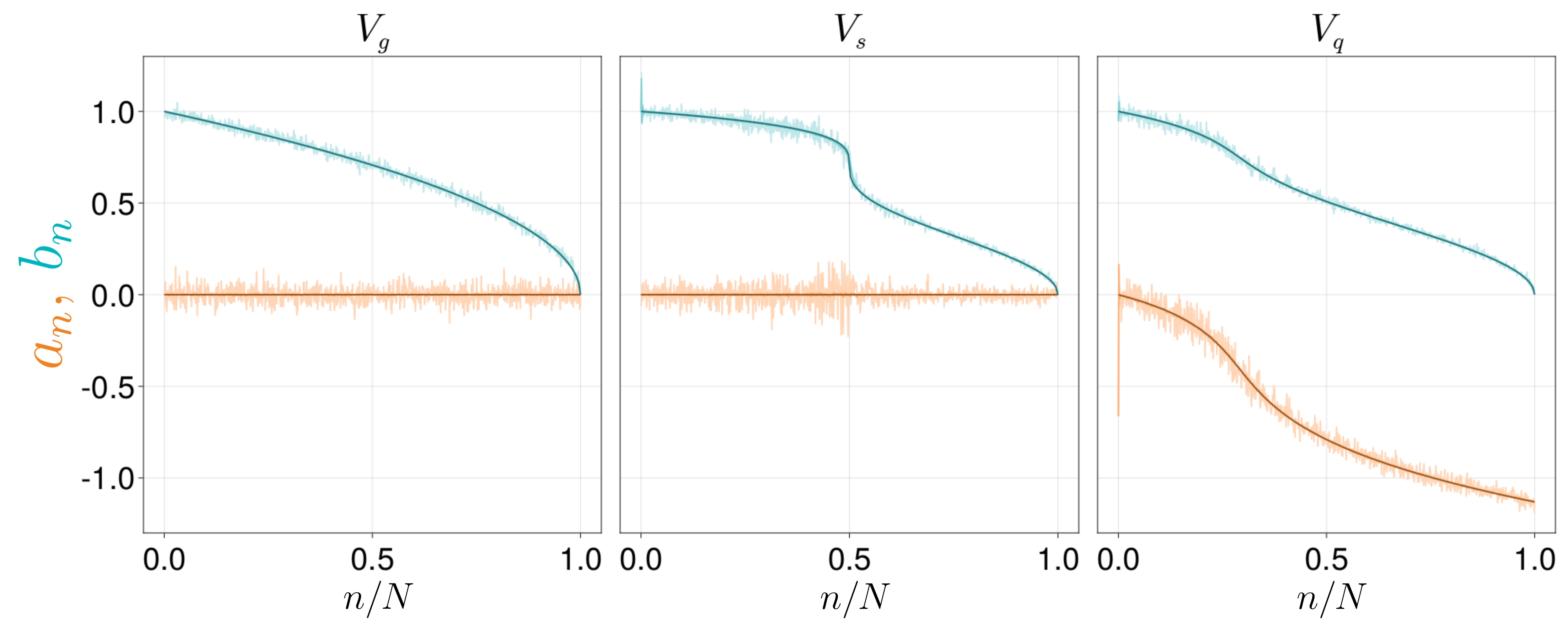}
\caption{Graph of the Lanczos coefficients $a(x),b(x)$ for one instance (light colors) and an average over $256$ instances (moderate colors) of size $N=1024$ random matrices with potentials $V_g$, $V_s$ and $V_q$ from left to right, along with the analytical solution for the average values (dark colors, continuous). The analytical solution overlaps well with the averaged values.
}
\label{fig:ex_ab}
\end{figure}

Next, consider a theory with a density of states
\be 
\rho(E)=\frac{\sqrt{4-E^2}}{\pi}\gr{\frac{7}{10}-\frac{3}{5}E^2+\frac{1}{5}E^4}\;.
\ee
Using the Coulomb Gas method as explained in \cite{https://doi.org/10.48550/arxiv.1510.04430}, the principal value integral 
\be
\frac{1}{4}V'(\omega)= \text{p.v.}\int dE\, \frac{\rho(E)}{\omega-E}\;,
\label{eq:pvintegral}
\ee
leads to the associated potential as
\be 
V_s(E)\equiv 3E^2-E^4+\frac{2}{15}E^6\;.
\ee
Using this density of states as the input to the integral equation (\ref{eq:intg_eqn_abrho}) we numerically obtain the average Lanczos coefficients using the algorithm discussed above. The results agree well with the average over the Lanczos coefficients computed for individual draws from the ensemble (middle plot in Fig~(\ref{fig:ex_ab})).

The previous two examples had densities of states that are even in $E$, giving vanishing $a$-type Lanczos coefficients as discussed above. To obtain non-trivial $a$ coefficients we consider a theory with a non-even density of states \be 
\rho(E)=\frac{\sqrt{4-E^2}}{\pi}\gr{\frac{1}{3}-\frac{1}{3}E+\frac{1}{6}E^2}\;.
\ee
Following (\ref{eq:pvintegral}),  this arises in a matrix model with a quartic potential including a cubic term
\be 
V_q(E)\equiv \frac{1}{6}E^4-\frac{4}{9}E^3+\frac{8}{3}E\;.
\ee
Using this density of states in the integral equation (\ref{eq:intg_eqn_abrho}) we again compute the average Lanczos coefficients. The results are in precise agreement with the ensemble average (right-hand plot in Fig.~(\ref{fig:ex_ab})).

\section{Statistics of the Lanczos coefficients.}


The starting point for studying the spectral statistics of a Random Matrix Theory is the joint probability distribution for the eigenvalues~(\ref{jointb}). In this section, we will similarly obtain the joint probability distribution of the Lanczos coefficients. We first recall from above that for a Gaussian random matrix these quantities are independent random variables~\eqref{indl}. Therefore, multiplying the probability distribution for each random variable, the joint  distribution of the Lanczos coefficients for the Gaussian generalized $\beta$-ensembles can be written as
\be
    p_{\textrm{Gaussian}}(a_0,\cdots ,a_{N-1},b_1,\cdots ,b_{N-1}) \propto \gr{\prod_{n=0}^{N-1} e^{-\beta N\frac{a_n^2}{4}}} \gr{\prod_{n=1}^{N-1} b_n^{(N-n)\beta-1} e^{-\beta N\frac{b_n^2}{2}}}\;.
\ee
This can be generalized to random matrix ensembles with arbitrary potentials $V(H)$. As mentioned above these are defined by a modified measure
\be
    \frac{1}{Z_{\beta,N, V}}e^{-\frac{\beta N}{4}\textrm{Tr}(V(H))}\;.
\ee
Since $\textrm{Tr}(V(H))$ does not change under unitary transformations, the joint distribution of the Lanczos coefficients becomes
\bea
    p(a_0,\cdots ,a_{N-1},b_1,\cdots ,b_{N-1}) &\propto & p_{\textrm{Gaussian}}(a_0,\cdots ,a_{N-1},b_1,\cdots ,b_{N-1})\, \frac{e^{-\frac{\beta N}{4}\textrm{Tr}(V(H))}}{e^{-\frac{\beta N}{4}\textrm{Tr}(H^2)}}\nonumber\\ &\propto & \gr{\prod_{n=1}^{N-1} b_n^{(N-n)\beta-1}}\,e^{-\frac{\beta N}{4}\textrm{Tr}(V(H))}\;.\label{eq:joint_prob_dist}
\eea
Here we used the fact that $\textrm{Tr}(H^2)$ for the triagonalized Hamiltonian is just $\sum_n (a_n^2 + 2 b_n^2)$.
Equivalently, the Jacobian of the coordinate transformation from the original form of the matrix to its tridiagonal form via the Lanczos procedure is
\be 
J\propto  \prod_{n=1}^{N-1} b_n^{(N-n)\beta-1}\;.\label{eq:Jacobian}
\ee
Since this coordinate change applies to single realizations of the random matrix, once we find it for one example, say the Gaussian ensemble, it will be the same for any other distribution. The same happens when we change variables to the eigenvector basis, with the universal appearance of the Vandermonde determinant for any given potential.

This Jacobian can also be obtained directly by finding the volume of the space of matrices that correspond to a particular set of Lanczos coefficients. Recall that every step of the Householder transformation procedure of tridiagonalizing a matrix converts matrices as follows:
\eqm{\begin{pmatrix}
\ddots&\ddots&&\\
\ddots&a_{n-1}&b_n&\\
&b_n&a_n & x^{T} \\
&&x & H_{N-n-1} 
\end{pmatrix}\to \begin{pmatrix}
\ddots&\ddots&&\\
\ddots&a_{n-1}&b_n&\\
&b_n&a_n & b_{n+1}e_1^T \\
&&b_{n+1}e_1 & H_{N-n-1}'
\end{pmatrix}\;.
}
At every step, $(N-n-1)\beta$ independent real numbers corresponding to $\beta$ parameters for each of the $(N-n-1)$ entries of $x$ are collapsed into their magnitude, $b_{n+1}$. The space of entries corresponding to the same $b_{n+1}$ are those on a surface of a $(N-n-1)\beta$ dimension sphere with radius $b_{n+1}$, giving us a factor of $b_{n+1}^{(N-n-1)\beta-1}$ in the volume. Multiplying the volume factor due to every step of the transformation, we find that a total volume proportional to 
\be
\prod_{n=1}^{N-1} b_{n}^{(N-n)\beta-1}\;,
\ee
in the original space of matrices corresponds to the same tridiagonal matrix after the Lanczos procedure.

For a third approach, we can follow the results of \cite{Dumitriu_2002}, where the Jacobian that takes us from the tridiagonal form of GOE to the eigenvalue form was found to be
\be 
J_{T\rightarrow \lambda}\propto\frac{\Delta}{\prod_{n=1}^{N-1} b_n^{N-n-1}}\;,
\ee
where $\Delta$ is the Vandermonde determinant. The proof of the form of that Jacobian does not depend on the specific potential we assume for the orthogonal matrices. Therefore we can invert that Jacobian and multiply it by the Jacobian that takes us from the original matrix form to the eigenvalue form, namely the Vandermonde determinant $\Delta$. We again arrive at the result \eqref{eq:Jacobian}.

\subsection{Saddle point approach to the one-point function}

At large $N$, the probability distribution of the Lanczos coefficients \eqref{eq:joint_prob_dist} becomes peaked around its average value, so we may use a saddle point approximation to find the average and covariance of the Lanczos coefficients. The logarithm of the probability from the exponent of \eqref{eq:joint_prob_dist} is given by 
\be
S_{\textrm{eff}}\equiv \ln p(a_0,\ldots, a_{n-1},b_1,\ldots,b_{N-1}) =\sum_n \gr{(N-n)\beta-1}\ln b_n - \frac{\beta N}{4}\textrm{Tr}(V(H))\;.
\ee
The average coefficients, or one-point functions, are those that maximize $S_\textrm{eff}$.  We can find them by taking the gradient and setting it to zero. In this sense, the function $S_{eff}$ plays the role of an effective action for the Lanczos coefficients.

We want to take derivatives of $\textrm{Tr}(V(H))$ with respect to the Lanczos coefficients and to evaluate them at the average of the distribution.  To evaluate derivatives with respect to a given $a_i$ or $b_i$ we will see that it is convenient to first expand the potential as a polynomial around the average value of $a_i$, namely $V(E)=\sum_n w_n(E-\bar{a}_i)^n$. First, we recall the  matrix identity
\be
\frac{\dx}{\dx A_{ij}} \textrm{Tr}(A^n)=n\sq{A^{n-1}}_{ji}\;,
\ee
Defining $\bar{a}_i$ to be the average value of $a_i$, it follows from this identity that
\be
\frac{\dx}{\dx H_{ij}} \textrm{Tr}((H-\bar{a}_iI)^n)=n\sq{(H-\bar{a}_iI)^{n-1}}_{ji}\;,
\ee
where $I$ is the identity matrix. As discussed in the previous section, when $i,j$ are close to $k$, the same continuity assumption that we used to derive \eqref{eq:local_approx} implies that the average Hamiltonian is then  approximated by
\be \label{eq:local_approx_2}
\sq{(\bar{H}-\bar{a}_kI)^n}_{ij}\approx \sq{T(0,\bar{b}_k)^n}_{ij}\;,
\ee
where $T(a,b)$ was defined above in~(\ref{Triav}) and $I$ is the identity matrix. Also note that, as before, the approximation is not accurate when $i,j$ are $O(1)$ in the large $N$ limit. Next, we notice that
\bea
    \frac{\dx}{\dx b_i} \textrm{Tr}((H-\bar{a}_i)^n)
    &=&2 \frac{\dx}{\dx H_{i-1,i}} \textrm{Tr}((H-\bar{a}_i)^n)
    = 2n\sq{(H-\bar{a}_iI)^{n-1}}_{i,i-1}\;,\\
    \frac{\dx}{\dx a_i} \textrm{Tr}((H-\bar{a}_i)^n)
    &=&\frac{\dx}{\dx H_{i,i}} \textrm{Tr}((H-\bar{a}_i)^n)
    =n\sq{(H-\bar{a}_iI)^{n-1}}_{i,i}\;.
\eea
We now evaluate at the average of the tridiagonal Hamiltonian and use \eqref{eq:trid_to_binomial} to arrive at
\bea
    \ev{\frac{\dx}{\dx b_i} \textrm{Tr}((H-\bar{a}_i)^n)}_{H=\bar{H}}&=& 2n\sq{T(0,\bar{b}_k)^{n-1}}_{i,i-1}= 2n\bar{b}_i^{n-1}\binom{n-1}{n/2}\;,\\
    \ev{\frac{\dx}{\dx a_i} \textrm{Tr}((H-\bar{a}_i)^n)}_{H=\bar{H}}&=&n\sq{T(0,\bar{b}_k)^{n-1}}_{i,i}=
    n\bar{b}_i^{n-1}\binom{n-1}{(n-1)/2}    \;.
\eea
Here we see why it was convenient to expand the potential around $\bar{a}_1$: this allowed us to exploit the binomial identity for powers of $T(0,\bar{b}_k)$.

We now come back to the generic potential written as $V(E)=\sum_n w_n(E-\bar{a}_i)^n$. Using the integral identities for the binomial coefficients in \eqref{eq:trid_integral} we can write the derivatives as follows
\bea
    \ev{\frac{\dx}{\dx b_i} \textrm{Tr}(V(H))}_{H=\bar{H}}&=&2\sum_n n w_n \bar{b}_i^{n-1}\binom{n-1}{n/2} = \ev{\gr{\frac{\dx}{\dx b_i}\int_{-2}^2 dx\,\frac{V(a_i+b_ix)}{\pi\sqrt{4-x^2}}}}_{H=\bar{H}}\;,
    \nonumber\\
      \ev{\frac{\dx}{\dx a_i} \textrm{Tr}(V(H))}_{H=\bar{H}}&=&\sum_n n w_n \bar{b}_i^{n-1}\binom{n-1}{(n-1)/2} =\ev{\gr{\frac{\dx}{\dx a_i}\int_{-2}^2 dx\,\frac{V(a_i+b_ix)}{\pi\sqrt{4-x^2}}}}_{H=\bar{H}}\;,\label{eq:one_point_vdiffs}
\eea
where we used the observations that 
\bea
\frac{\dx}{\dx b_i}V(a_i+b_ix)\vert_{H=\bar{H}} &=&xV'(a_i+b_ix)\vert_{H=\bar{H}} =\sum_n  n w_n \bar{b}_i^{n-1} x^{n}
\;,\\
\frac{\dx}{\dx a_i}V(a_i+b_ix)\vert_{H=\bar{H}} &=&V'(a_i+b_ix)\vert_{H=\bar{H}} =\sum_n  n w_n \bar{b}_i^{n-1}x^{n-1}
\;.
\eea
Changing variables to $E=a_i+b_i x$, we finally obtain
\bea
    \frac{\dx}{\dx a_i} \textrm{Tr}(V(H)) &=& \frac{\dx}{\dx a_i}\int dE\,\frac{V(E)}{\pi\sqrt{4b_i^2-(E-a_i)^2}}\;,\nonumber\\
    \frac{\dx}{\dx b_i} \textrm{Tr}(V(H)) &=& \frac{\dx}{\dx b_i}\int dE\,\frac{V(E)}{\pi\sqrt{4b_i^2-(E-a_i)^2}}\;.\label{eq:one_point_simplified}
\eea
where the limits of the integration in the energy are fixed by those of $x$.
These last expressions can also be derived using the integral equation relating the density of states and the Lanczos coefficients~(\ref{intdl}).\footnote{In more detail, we can substitute the integral expression for the density of states (\ref{intdl}) into ${\rm Tr}(V(H)) = \int dE \, \rho(E) V(E)$.  The integral expression for $\rho$ involves a sum over all the Lanczos coefficients.  Taking the derivative as in (\ref{eq:one_point_simplified}) pulls out one term in the sum.} Conversely, this derivation provides a different path to the integral equation. 

Integrating both sides of (\ref{eq:one_point_simplified}) gives an integral expression for ${\rm Tr}(V(H))$.  Thus, the effective action can be written as
\be
  S_{\textrm{eff}} = \sum_n \gr{(N-n)\beta-1}\ln b_n - \frac{\beta N}{4} \sum_n \int dE \frac{V(E)}{\pi \sqrt{4b_n^2 - (E-a_n)^2}}\;.
\ee
Writing the Lanczos coefficients as a function of $x=n/N$ (not to be confused with the $x$ in (\ref{eq:one_point_vdiffs})),  namely $a(x)$ and $b(x)$, we can rewrite the effective action in the large-$N$ limit as
\be
  \frac{S_{\textrm{eff}}}{\beta N^2} = \int dx\,(1-x)\ln b(x) - \frac{1}{4} \int dx\int dE \frac{V(E)}{\pi \sqrt{4b(x)^2 - (E-a(x))^2}}\;.
\ee
This form of the effective action (the logarithm of the joint probability distribution of the Lanczos coefficients) shows that the nice variables in the large-$N$ limit are $a(x)$ and $b(x)$, as assumed in the previous section. To maximize the probability we just need to find the extrema of this action. We obtain the following coupled equations for the Lanczos coefficients
\bea
    4(1-x) &=& b(x) \frac{\dx}{\dx b(x)}\gr{\int dE \frac{V(E)}{\pi \sqrt{4b(x)^2 - (E-a(x))^2}}}\;,\nonumber\\
    0&=&\frac{\dx}{\dx a(x)}\gr{\int dE \frac{V(E)}{\pi \sqrt{4b(x)^2 - (E-a(x))^2}}}\;.\label{eq:one_point_final}
\eea
If $V(E)=\sum_n w_nE^n$ is a polynomial, then the integrals also give polynomials, resulting in a system of algebraic equations in $a,b,x$. This system reads
\bea
    4(1-x) &=& \sum_n w_n\sum_m m a^{n-m}b^{m-1}\binom{n}{m}\binom{m}{m/2} \;,\nonumber\\
    0&=&\sum_n w_n\sum_m (n-m)a^{n-m-1}b^{m}\binom{n}{m}\binom{m}{m/2}\;.\label{eq:one_point_poly}
\eea
These equations are consistent with the generic equations derived in Ref.\cite{doi:10.1063/1.533010,PhysRevD.47.1640} for the thermodynamic limit of the Lanczos method. We have arrived at them in a simpler manner, using conventional saddle point techniques in the context of Random Matrix Theory.

In the previous equation we can see that if $V(E)$ is an even polynomial, then the second equation can be solved with $a=0$, since the only terms that contribute have even $m$ so every term in the sum has a factor of $a$. This allows us to write $x$ as a polynomial of $b$.

We now verify that the solutions to the integral equation in the examples considered previously can also be obtained by solving this system of equations. For the Gaussian case, namely $V_g(E)=E^2$ the equations are just
\be 
a(x)=0\;,\,\,\,\,\,\,\,\,\,\,\,\,\, b(x)^2=(1-x)\;.
\ee
For the second potential $V_s(E)=3E^2 - E^4 + \frac{2}{15}E^6$ we find
\be 
a=0\;,\,\,\,\,\,\,\,\,\,\,\,\,\,  3b^2-6b^4+4b^6=(1-x)\;.
\ee
This algebraic equation can be solved numerically and the solution precisely matches the result coming from the integral equation \eqref{intdl} and displayed in Fig.~\ref{fig:ex_ab}. For the last case, namely $V_q(E)= \frac{1}{6}E^4 - \frac{4}{9}E^3 + \frac{8}{3}E$, we observe that the Lanczos coefficients obey the following equations 
\bea
0&=&\frac{1}{24}(4a^3+24ab^2)-\frac{1}{9}(3a^2+6b^2)+\frac{2}{3}\;,\\
(1-x)&=&\frac{1}{24}(24a^2b^2+24b^4)-\frac{1}{9}(12ab^2)\;.
\eea
This system of equations has multiple solutions $a,b$ at some values of $x$, but choosing the solution with the smallest $b$, we find the numerical solution again coincides precisely with the one obtained by solving the integral equation relating the Lanczos coefficients and the density of states, and is shown in Fig.~\ref{fig:ex_ab}.

These results establish a connection between solutions of the integral equation for the Lanczos coefficients, and algebraic relations arising from the saddlepoint conditions. It would be interesting to better understand the origin of this connection.

\subsection{Saddle point approach to the two-point function}

To find the covariance of the Lanczos coefficients, or, equivalently, their two-point function, we approximate the probability distribution as a Gaussian around its peak. This is equivalent to expanding the effective action to quadratic order around the saddle point derived above.\footnote{In physics language, we are expanding the action around the classical solution.} We thus seek an expansion
\be
    S_{eff}(a,b) = S_{eff}(\bar{a},\bar{b})+\Delta S_{\textrm{eff}}\;,
\ee
where 
\be 
\Delta S_{\textrm{eff}}\equiv -\frac{1}{2}\,(\, \delta a_i\,M_{ij}^{aa}\,\delta a_i+2\delta a_i\,M_{ij}^{ab}\,\delta b_i+\delta b_i\,M_{ij}^{bb}\,\delta b_i\,)\;,
\ee
and we have defined the Gaussian kernels
\bea
M_{ij}^{aa}&=&\frac{\beta N}{4}\frac{\dx^2}{\dx a_i \dx a_j}\textrm{Tr}(V(H))\;,\nonumber\\
M_{ij}^{bb}&=&\frac{\beta N}{4}\frac{\dx^2}{\dx b_i \dx b_j}\textrm{Tr}(V(H))-\frac{\dx^2}{\dx b_i \dx b_j}\sum_n ((N-n)\beta-1)\ln b_n\;,\nonumber\\
M_{ij}^{ab}=M_{ji}^{ba}&=&\frac{\beta N}{4}\frac{\dx^2}{\dx a_i \dx b_j}\textrm{Tr}(V(H))\;.\label{eq:M_definitions}
\eea
The two-point functions of the Lanczos coefficients are computed from the inverse of these kernels.

The first step is to compute the second derivative of $\textrm{Tr}(V(H))$ with respect to the Lanczos coefficients. As above, we do this by expanding $V(H)$ as a polynomial and then expressing the sum in an alternate form. For a general matrix $H$, a standard identity gives the second derivative of $\textrm{Tr}\,H^n$ with respect to the matrix entries $H_{ij}$ as
\be
    \frac{\dx^2}{\dx H_{ij} \dx H_{kl} } \textrm{Tr}H^n= n \sum_{m=0}^{n-2} \gr{H^{m}}_{li}\gr{H^{n-2-m}}_{jk}\;.
    \label{eq:SecondDerivIdent}
\ee
We want to apply this formula in the basis in which $H$ is tridiagonal, and then evaluate the result at the average value of the Lanczos coefficients. 

First notice that, using (\ref{eq:SecondDerivIdent}), 
\bea
    \frac{\dx^2}{\dx a_i \dx a_j } \textrm{Tr}(H-\bar{a}_k)^{n}
    &=& \frac{\dx^2}{\dx H_{i,i} \dx H_{j,j} } \textrm{Tr}(H-\bar{a}_k)^{n} = n \sum_{m=0}^{n-2} \sq{(H-\bar{a}_k)^{m}}_{j,i}\sq{(H-\bar{a}_k)^{n-2-m}}_{i,j}\;,\nonumber
    \\
    \frac{\dx^2}{\dx a_i \dx b_j } \textrm{Tr}(H-\bar{a}_k)^{n} 
    &=& 2\frac{\dx^2}{\dx H_{i,i} \dx H_{j-1,j} } \textrm{Tr}(H-\bar{a}_k)^{n} = 2n \sum_{m=0}^{n-2} \sq{(H-\bar{a}_k)^{m}}_{ji}\sq{(H-\bar{a}_k)^{n-2-m}}_{i,j-1}\;,\nonumber\\
    \frac{\dx^2}{\dx b_i \dx b_j } \textrm{Tr}(H-\bar{a}_i)^{n}  &=& 2\frac{\dx^2}{\dx H_{i-1,i} \dx H_{j-1,j} } \textrm{Tr}(H-\bar{a}_k)^{n}+2\frac{\dx^2}{\dx H_{i,i-1} \dx H_{j-1,j} } \textrm{Tr}(H-\bar{a}_k)^{n}\nonumber\\&=&
    2n \sum_{m=0}^{n-2} \sq{(H-\bar{a}_k)^{m}}_{j,i-1}\sq{(H-\bar{a}_k)^{n-2-m}}_{i,j-1}\nonumber\\
    &&+2n \sum_{m=0}^{n-2} \sq{(H-\bar{a}_k)^{m}}_{j,i}\sq{(H-\bar{a}_k)^{n-2-m}}_{i-1,j-1}\;.
\eea
As we did for the one-point functions, we evaluate these expressions at the average value of the Lanczos coefficients in the large $N$ limit. In this limit, choosing $k$ close to $i,j$ allows us to use \eqref{eq:local_approx_2} and \eqref{eq:trid_to_binomial}. We get
\bea
    \ev{\frac{\dx^2}{\dx a_i \dx a_j } \textrm{Tr}(H-\bar{a}_i)^{n}}_{H=\bar{H}} &\approx&n \sum_{m=0}^{n-2} \sq{T(0,\bar{b}_k)^{m}}_{ji}\sq{T(0,\bar{b}_k)^{n-2-m}}_{ij}
    = n\bar{b}_k^{n-2} C^{n-2}_{j-i,j-i}\;,\nonumber\\
    \ev{\frac{\dx^2}{\dx a_i \dx b_j } \textrm{Tr}(H-\bar{a}_i)^{n}}_{H=\bar{H}} &\approx& 2n \sum_{m=0}^{n-2} \sq{T(0,\bar{b}_k)^{m}}_{ji}\sq{T(0,\bar{b}_k)^{n-2-m}}_{i,j-1}
    =2n \bar{b}_k^{n-2} C^{n-2}_{j-i,j-i-1} \;,\nonumber\\
    \ev{\frac{\dx^2}{\dx b_i \dx b_j } \textrm{Tr}(H-\bar{a}_i)^{n}}_{H=\bar{H}} &\approx&
    2n \sum_{m=0}^{n-2} \sq{T(0,\bar{b}_k)^{m}}_{j,i-1}\sq{T(0,\bar{b}_k)^{n-2-m}}_{i,j-1}\nonumber\\
    &&+2n \sum_{m=0}^{n-2} \sq{T(0,\bar{b}_k)^{m}}_{j,i}\sq{T(0,\bar{b}_k)^{n-2-m}}_{i-1,j-1}\nonumber\\
    &=&2n\bar{b}_k^{n-2} \gr{C^{n-2}_{j-i+1,j-i-1} +C^{n-2}_{j-i,j-i}}\;,\label{eq:trace_second_derivatives}
\eea
where we have defined
\be
C_{\alpha,\beta}^n \equiv \sum_{k=0}^{n} \binom{k}{(k+\alpha)/2}\binom{n-k}{(n-k+\beta)/2}\;.
\ee
As before we always assume the binomials are zero when their arguments are not integers.

In appendix \ref{appendix_math}, we prove some properties satisfied by the coefficients $C_{\alpha,\beta}^n$. In particular,  $C_{\alpha,\beta}^n=C^n_{0,\alpha+\beta}$ for $\alpha,\beta\geq 0$, so we will write everything in terms of $C_{\delta}^n\equiv C_{0,\delta}^n$. Using this relation, and choosing $k=i$ and $j=i+\delta$ we can rewrite the second derivatives \eqref{eq:trace_second_derivatives} as
\bea\label{third}
    \ev{\frac{\dx^2}{\dx a_i \dx a_{i+\delta} } \textrm{Tr}(H-\bar{a}_i)^{n}}_{H=\bar{H}} &\approx& n\bar{b}_i^{n-2} C^{n-2}_{\abs{2\delta}}\;,\nonumber\\
    \ev{\frac{\dx^2}{\dx a_i \dx b_{i+\delta} } \textrm{Tr}(H-\bar{a}_i)^{n}}_{H=\bar{H}} &\approx& 2n \bar{b}_i^{n-2} C^{n-2}_{\abs{2\delta-1}}\;, \nonumber\\
    \ev{\frac{\dx^2}{\dx b_i \dx b_{i+\delta} } \textrm{Tr}(H-\bar{a}_i)^{n}}_{H=\bar{H}} &\approx& 2n\bar{b}_i^{n-2} \gr{C^{n-2}_{\delta+1,\delta-1} +C^{n-2}_{\abs{2\delta}}}\;.
\eea
From now on, we rename $\bar{a}_i\rightarrow a_i$ and $\bar{b}_i\rightarrow b_i$; there will be no ambiguity because we will be evaluating all quantities at the average values of the Lanczos coefficients.

In the third equation of~(\ref{third}), when $\delta\neq 0$, we can simplify $C^{n-2}_{\delta+1,\delta-1}=C^{n-2}_{\abs{2\delta}}$, but $\delta=0$ needs to be treated as a special case since $C^{n-2}_{1,-1}=C^{n-2}_{2}$. Ignoring that special case for now, if we expand $V(E)=\sum_n w_n (E-\bar{a}_i)^n$, we can compute the Gaussian kernels $M$ from \eqref{eq:M_definitions} as
\bea
    M_{i,i+\delta}^{aa}&\approx&\frac{\beta N}{4}\sum_n nw_n b_i^{n-2} C^{n-2}_{\abs{2\delta}} \;,\nonumber\\
    M_{i,i+\delta}^{bb}&\approx&\frac{\beta N}{4}\sum_n 4nw_n b_i^{n-2} C^{n-2}_{\abs{2\delta}}\;, \nonumber\\
    M_{i,i+\delta}^{ab}=M_{i+\delta,i}^{ba}&\approx&\frac{\beta N}{4}\sum_n 2nw_n b_i^{n-2} C^{n-2}_{\abs{2\delta-1}}\;.\label{eq:M_expanded}
\eea
Here we have used that $-\frac{\dx^2}{\dx b_i \dx b_{i+\delta}}\sum_n ((N-n)\beta-1)\ln b_n=0$ when $\delta\neq0$.  This is an explicit algebraic expression for the Gaussian kernels associated with any polynomial potential.

In the $\delta=0$ case, we have 
\bea
    M_{i,i}^{bb}&\approx&\frac{\beta N}{4}\sum_n 2nw_n b_i^{n-2} \gr{C^{n-2}_{2}+C^{n-2}_{0}}+ \frac{(N-i)\beta-1}{b_i^2}\;.
    \label{eq:Delta0Eq}
\eea
This can be further simplified. From \eqref{pascalconv_recursionrelations}, we derive that 
\be
    C_2^n=C_{0}^n-\frac{1}{2}\binom{n+2}{(n+2)/2}\;.
\ee
Meanwhile, we can use the saddle point equations (\ref{eq:one_point_vdiffs},\ref{eq:one_point_simplified},\ref{eq:one_point_final}), to rewrite the last quantity in  (\ref{eq:Delta0Eq}) in terms of 
\be
    4\frac{(N-i)\beta-1}{\beta N} \approx  4(1-x) \approx 2b_i\sum_n n w_n b_i^{n-1}\binom{n-1}{n/2}=\sum_n n w_n b_i^{n}\binom{n}{n/2}\;.
\ee
The second expression follows by defining $x = i/N$ and taking the large $N$ limit.  The last expression follows by using Pascal's identity for binomials. Therefore
\bea
    M_{i,i}^{bb}&\approx&\frac{\beta N}{4}\sum_n 4nw_n b_i^{n-2} C^{n-2}_{0}-\frac{\beta N}{4}\sum_n nw_n b_i^{n-2}\binom{n}{n/2}+\frac{\beta N}{4}\frac{4(1-x)}{b_i^2}\nonumber\\
    &=&\frac{\beta N}{4}\sum_n 4nw_n b_i^{n-2} C^{n-2}_{0}\;.
\eea
So \eqref{eq:M_expanded} is true for $\delta=0$ as well.

These expressions determine the Gaussian kernel for polynomial potentials $V(E)$. In the interest of arriving at an analytical treatment for general $V(E)$ (or general density of states), we can convert these sums to integral expressions in the large $N$ limit. To do so, we use results from appendix~(\ref{appendix_math}). In particular, as shown in Eqs.~(\ref{eq:rhom_definition} and~(\ref{eq:Pm_definition}),
we can write the $C_m^n$ as  moments $\int x^n \, \eta_m(x) \, dx$ of
\bea
    \eta_m(x) &=& \frac{1}{2}\delta(x-2)+(-1)^m\frac{1}{2}\delta(x+2)+\frac{1}{2}\frac{x P_m(x)}{\pi \sqrt{4-x^2}}H(4-x^2)\;,\nonumber \\
    P_m(x)&=&-\frac{\sin\gr{m\cos^{-1}\gr{\frac{x}{2}}}}{\sin\gr{\cos^{-1}\gr{\frac{x}{2}}}}\;.
    \label{eq:EtaAndP}
\eea
In what follows, we will focus on the correlations of $a_i, b_i$ with $a_j, b_j$ when  $|i-j| = \delta$ with $\delta/N \to 0$ in the large $N$ limit.  We focus on this regime because the correlations fall off with $\delta$ in the large $N$ limit.
Then, in terms of the functions (\ref{eq:EtaAndP}) the Gaussian kernels become
\bea
    M_{i,i+\delta}^{aa}&\approx&\frac{\beta N}{4}\int dx~\sum_n n\,w_n \,(b_ix)^{n-2}\,\eta_{\abs{2\delta}}(x)
    \nonumber\\
    &=&\frac{\beta N}{4}\int dE~\frac{V'(E)-V'(a_i)}{b_i(E-a_i)}\,\eta_{\abs{2\delta}}\gr{\frac{E-a_i}{b_i}}\;,\nonumber\\
    M_{i,i+\delta}^{bb}&\approx&\beta N\int dE~\frac{V'(E)-V'(a_i)}{b_i(E-a_i)}\,\eta_{\abs{2\delta}}\gr{\frac{E-a_i}{b_i}} \;,\nonumber\\
    M_{i,i+\delta}^{ab}=M_{i+\delta,i}^{ba}&\approx&\frac{\beta N}{2}\int dE~\frac{V'(E)-V'(a_i)}{b_i(E-a_i)}\,\eta_{\abs{2\delta-1}}\gr{\frac{E-a_i}{b_i}}\;.\label{eq:M_expanded_2}
\eea
Again, due to the continuity of the Lanczos coefficients in the large $N$ limit, and since $M^{aa}_{ij}$ is approximated by a function of $a_i,b_i$ that does not depend explicitly on $i$, we can make the approximation that for small $\delta$, $M_{ij}^{aa}\approx M_{i+\delta,j+\delta}$. This also holds for $M^{ab}_{ij},M^{bb}_{ij}$ and can be verified at numerically.

Finally, we want to invert the Gaussian kernel in the vicinity of some Lanczos index $i$.
To this end, rather than writing the full Gaussian kernel $M$ as a block matrix separating the $a$'s and $b$'s, we index $M$ as follows:
\be
    M_{2\alpha,2\beta}=M^{aa}_{i+\alpha,i+\beta}\;,\,\,\,\,\,\,\,\,\,\,\,\,\,\,\, M_{2\alpha-1,2\beta-1}=M^{bb}_{i+\alpha,i+\beta}\;,\,\,\,\,\,\,\,\,\,\,\,\,\,\,\, M_{2\alpha,2\beta-1}=M^{ab}_{i+\alpha,i+\beta}\;,
\ee
where $\alpha/N, \beta/N \to 0$ in the large $N$ limit. Then $M$ is the Gaussian kernel for both $a$ and $b$ Lanczos coefficients in the vicinity of the index $i$.  $M$ can be written as a product of matrices $SM'S$ where
\be\label{SMp}
    S_{pq}=\begin{cases}
    \delta_{pq},& p\;\text{even}\\
    2\delta_{pq},& p\;\text{odd}\\
    \end{cases},\;\;\;\;\;\;\;\;\;\;\;\; M'_{pq}=\int dE~\frac{V'(E)-V'(a_i)}{b_i(E-a_i)}\,\eta_{|p-q|}\gr{\frac{E-a_i}{b_i}}\;.
\ee
Notice that $M'$ is implicitly labeled by the Lanczos index $i$.
The matrix $S$ is diagonal, so its inverse is simple. The matrix $M'$ is  Toeplitz,  i.e., its entries are functions only of the distance to the diagonal $\vert p-q\vert$
in the large $N$ limit.
Thus the eigenvectors of $M'$ are plane waves $e^{ikq}$.  We can then compute its eigenvalues via
\bea
    \lambda_{M'}(k) &=& \sum_q \int dE~\frac{V'(E)-V'(a_i)}{b_i(E-a_i)}\,\eta_{\abs{q}}\gr{\frac{E-a_i}{b_i}}e^{ikq}\nonumber \\
    &=&\int dE~\frac{V'(E)-V'(a_i)}{b_i(E-a_i)}\,\eta\gr{\frac{E-a_i}{b_i},e^{ik}}\;.\label{eq:eig_defs}
\eea
Here we took the inner product of the matrix $M'$ with the eigenvector and its conjugate, and used the fact that the matrix is Toeplitz to reduce one integral to a delta function. In the second line, the sum over $q$ is brought into the integral via the function
\bea
    \eta(x,t)= \sum_{q=-\infty}^{\infty} \eta_q(x)\,t^q=\frac{x}{\pi\sqrt{4-x^2}}\frac{1}{t+\frac{1}{t}-x}\;,
\eea
derived in the appendix~(\ref{appendix_math}) (see Eq.~\ref{eq:rho_sum}).  We then invert $M'$ in the momentum basis and take the Fourier transform to find 
\be
M_{\alpha\beta}'^{-1}=\frac{1}{2\pi} \int_0^{2\pi}\frac{e^{ik(\alpha-\beta)}dk}{\lambda_{M'}(k)}\;.\label{eq:ift_integral}
\ee
Since the matrix $S$ defined in~(\ref{SMp}) is diagonal, this allows us to compute the inverse of the Gaussian kernel: 
\be
M^{-1}=S^{-1}\,M'^{-1}\,S^{-1}\;.
\ee
Thus, in conclusion, given $V(E)$, we may compute the integral \eqref{eq:eig_defs}, plug the resulting function into \eqref{eq:ift_integral}, and scale it by some factors of $2$ according to $S$ \eqref{SMp} to approximate the covariance of any two Lanczos coefficients in the vicinity of some index $i$, where by vecinity we mean any distance $\delta$ such  that $\delta/N$ dies in the large-N limit.

An immediate outcome of this computation is that, unlike the Gaussian case where the fluctuations of $a_n$ and $b_n$ are statistically independent, in a non-Gaussian random matrix theory $a_n$ and $b_n$ have non zero covariance. Another simple result is that since $S$ contains a factor of $2$ differentiating the entries associated with $a$'s and $b$'s, this equation implies the covariances of the Lanczos coefficients  satisfy
\be
  \overline{(a(x)-\bar{a}(x))(a(y)-\bar{a}(y))}=4\overline{(b(x)-\bar{b}(x))(b(y)-\bar{b}(y))}\;,
\ee
a relation valid in the large-N limit. A special case of this is the following relation between variances of the Lanczos coefficients
\be 
\sigma_{a(x)}^2=4\,\sigma_{b(x)}^2\;.
\ee

\subsubsection{Examples and numerical verification}
\label{sec:examples2}

We again consider the same three examples that we used to study the average Lanczos coefficients. The first of these was the GUE, defined by $V_g(E)=E^2$. The variance of the Lanczos coefficients can be obtained in this case from the exact tridiagonalization reviewed in a previous section. We now arrive at the same results using the saddle point approach.

For $V_g(E)=E^2$, eq.~\eqref{eq:M_expanded} gives us that the non-zero entries of $M$ are
\be
M^{aa}_{ii}=\frac{\beta N}{2}\;,\,\,\,\,\,\,\,\,\,\,\,\,\,\,M^{bb}_{ii}=2\beta N\;.
\ee
Thus the Lanczos coefficients are uncorrelated, with variance
\be 
\sigma_a^2=\frac{2}{\beta\, N}\;,\,\,\,\,\,\,\,\,\,\,\,\,\,\sigma^2_b=\frac{1}{2\beta\,N}\;.
\ee
This is consistent with the exact tridiagonalization in the large-N limit and compares well with the numerical evaluation in Fig~(\ref{fig:exnoise_ab}). Notice the factor of $4$ difference between the variances of $a$ and $b$.

\begin{figure}
\centering
\includegraphics[width=0.98\linewidth]{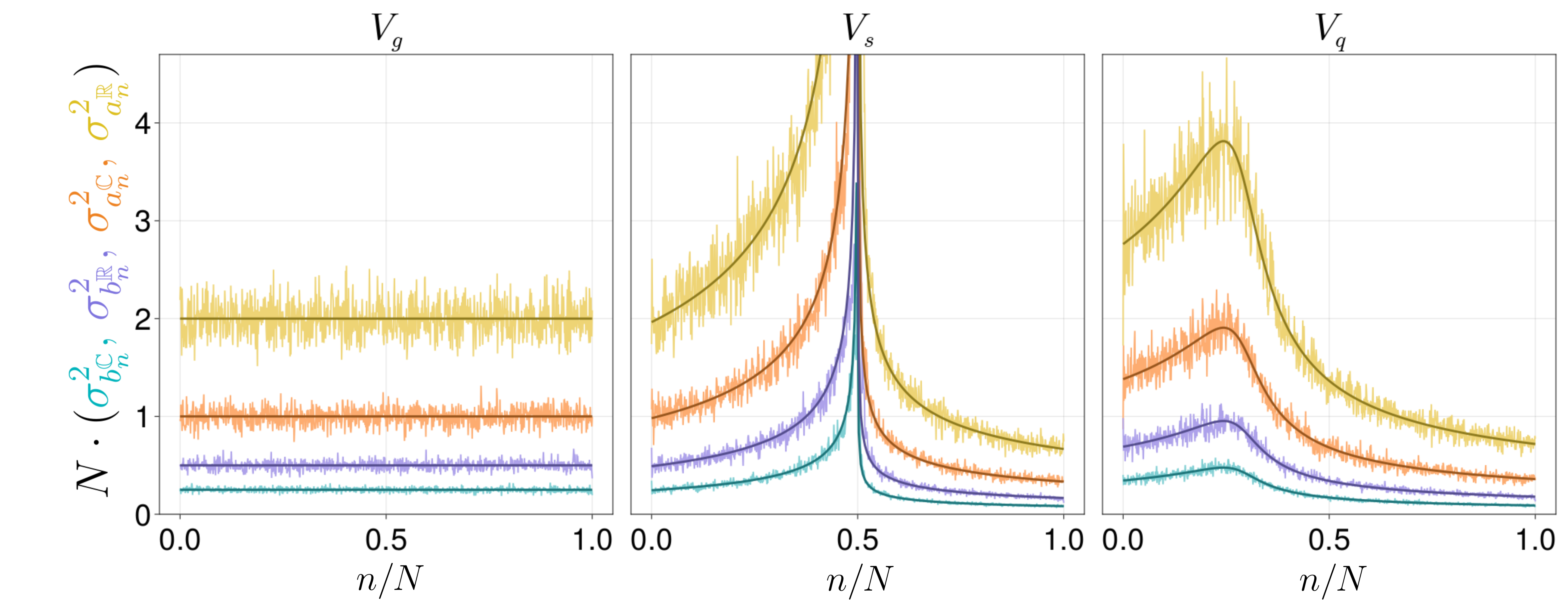}
\caption{Graph of the variance of $a(x),b(x)$ averaged over $256$ samples of $N=1024$ random matrices with potentials $V_g$, $V_s$ and $V_q$ from left to right (light colors), along with the analytical calculations of the variance (dark colors, continuous).
}
\label{fig:exnoise_ab}
\end{figure}

We now move to the potential
\be 
V_s(E)=3E^2-E^4+\frac{2}{15}E^6\;.
\label{eq:pot2}
\ee
In this case, \eqref{eq:M_expanded} gives   the only non-zero entries of $M$ as
\be
    M^{aa}_{i,i}=\beta N\gr{\frac{3}{2}-4b_i^2+\frac{16}{5}b_i^4}\;,\,\,\,\,\,\,\,\,\,\,\,\,\,\,M^{aa}_{i,i+1}=\beta N\gr{\frac{6}{5}b_i^4-b_i^2}\;,\,\,\,\,\,\,\,\,\,\,\,\,\,\, M^{aa}_{i,i+2}=\frac{\beta N}{5}b_i^4\;.
\ee
The values of $M^{bb}$ are four times the values for $M^{aa}$, and $M^{ab}=0$. The eigenvalues of $M'$ can be computed via
\be
    \lambda_{M'}(k)=M^{aa}_{i,i}+2M^{aa}_{i,i+1}\cos(2k)+2M^{aa}_{i,i+2}\cos(4k)\;.
\ee
Taking the inverse Fourier transform of the reciprocal as in \eqref{eq:ift_integral} gives us the entries of the inverse matrix. Numerically, one can take the matrix coefficients $M^{aa}_{i,i+\delta}$, pad it with an appropriate number of zeros for the desired accuracy, take the discrete Fourier transform via FFT (Fast Fourier Transform), take the reciprocals, and then invert the Fourier transform again via FFT to efficiently compute the entries of the inverse matrix.

Thus computing the eigenvalues directly leads to the two-point function. For the specific case of the variances, this result is plotted in the second panel of Fig~(\ref{fig:exnoise_ab}), along with the numerical evaluation of specific instances of random matrices generated with the potential (\ref{eq:pot2}). There is an excellent match between the analytic prediction and the average over this ensemble of random matrices.

As mentioned before, the previous two examples have zero average $a$-type Lanczos coefficients because they contain only even powers of the energy. To obtain non-trivial $a(x)$ we again consider the  potential
\be 
V_q(E)=\frac{1}{6}E^4-\frac{4}{9}E^3+\frac{8}{3}E\;.
\ee
We can compute the non-zero values of $M$, but this time as a function of both $a_i$ and $b_i$:
\be
    M^{aa}_{i,i}=\beta N\gr{ \frac{2}{3}b_i^2+\frac{1}{2}a_i^2-\frac{2}{3}a_i}\;,\,\,\,\,\,\,\,\,\,\,\,\,\,\,M^{aa}_{i,i+1}=\frac{\beta N}{6}b_i^2\;,\,\,\,\,\,\,\,\,\,\,\,\,\,\, M^{ab}_{i,i}=2\beta N\gr{\frac{1}{2}a_i-\frac{1}{3}b_i}\;.
\ee
These expressions lead to the eigenvalues
\be
\lambda_{M'}(k)=M^{aa}_{i,i}+M^{ab}_{i,i}\cos(k)+2M^{aa}_{i,i+1}\cos(2k)\;.
\ee
The variance of the Lanczos coefficients can then be computed as above. The result is plotted in the third panel of Fig~(\ref{fig:exnoise_ab}).

\section{Spectral form factor and TFD spread complexity}

The previous sections have determined the average and covariance of the Lanczos coefficients of a generic random matrix model. Now we want to show that this information by itself is sufficient to replicate many aspects of the long-time dynamics of chaotic systems. In particular, we first concentrate on the spectral form factor (SFF) defined by
\be 
\textrm{SFF}=\frac{Z_{\beta-i\,t}\,Z^*_{\beta+it}}{Z_{\beta}^2}\;,
\ee
where $Z_\beta=\sum_i\,e^{-\beta\,E_i}$ is the partition function of a  Hamiltonian with eigenvalues $E_i$. The SFF is well known in the context of matrix models, see \cite{Guhr:1997ve}, and has been recently studied in relation to black hole dynamics, see \cite{Cotler:2016fpe}.

\begin{figure}
    \centering
    \includegraphics[width=0.98\linewidth]{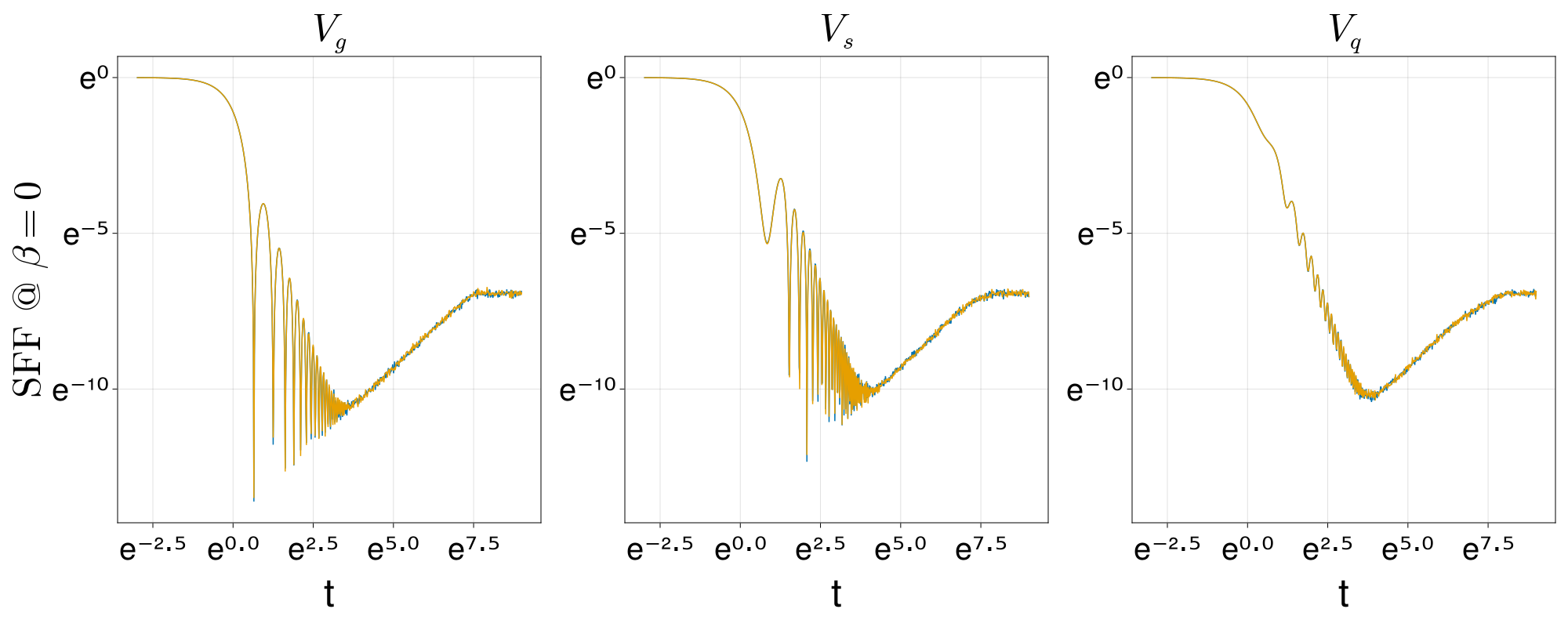}
    \caption{SFF for $N=1024$ random matrices distributed according to the potentials $V_g$, $V_s$, $V_q$, averaged over $256$ samples (in orange), as well as the average computed from $256$ samples based on the one and two-point functions of the Lanczos coefficients (in blue).}
    \label{fig:reverse_check_1}
\end{figure}

\begin{figure}
    \centering
    \includegraphics[width=0.98\linewidth]{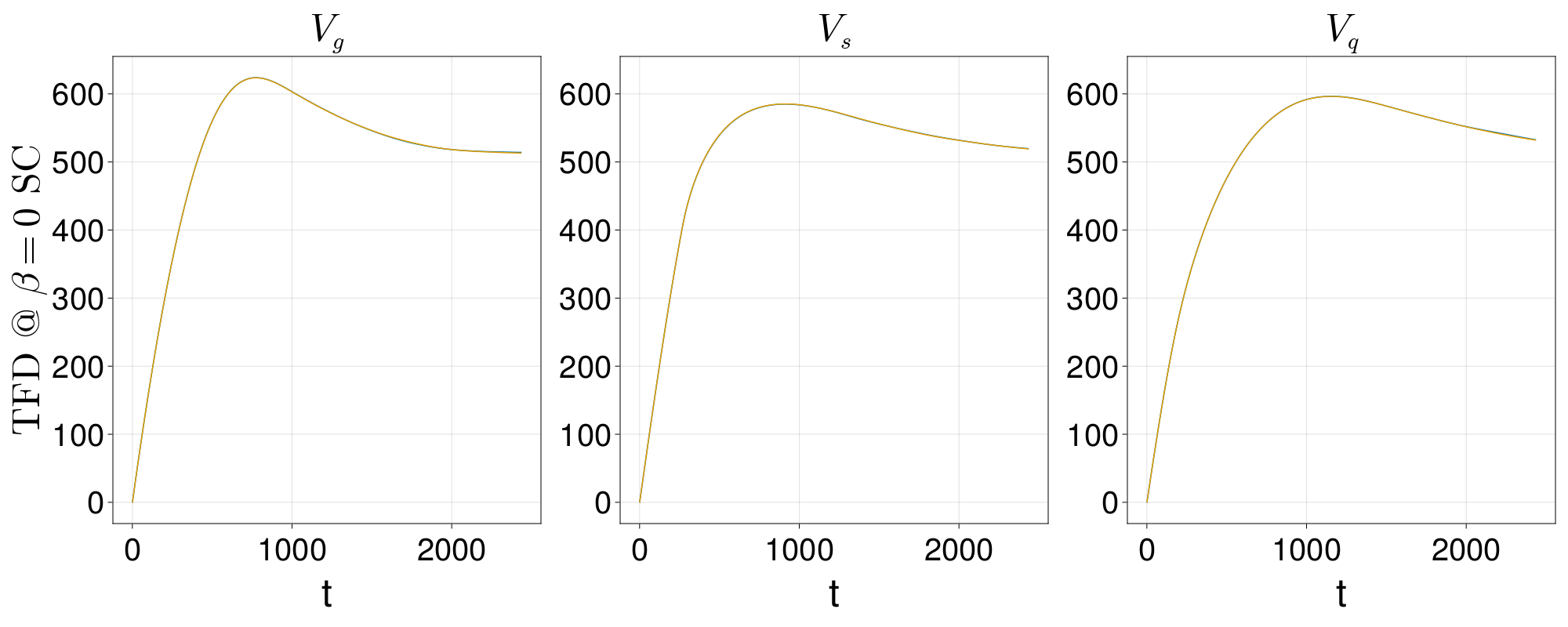}
    \caption{Spread complexity of the TFD state for $N=1024$ random matrices distributed according to the potentials $V_g$, $V_s$, $V_q$, averaged over $256$ samples (in orange), as well as the average computed from $256$ samples based on only the one and two-point functions of the Lanczos coefficients (in blue).}
    \label{fig:reverse_check_2}
\end{figure}

Second we consider the spread complexity proposed in \cite{SpreadC}, associated with the time evolution of the Thermo-Field Double (TFD) state.\footnote{The generic definition of spread complexity was reviewed above in~(\ref{LancSec}), see \eqref{eq:spread_comp_def}.} Concretely, the TFD state is defined as follows. Consider a Hamiltonian $H$, with eigenstates $\vert n \rangle$ and eigenvalues $E_n$. Then the TFD is defined by the state
\be 
\vert\psi_{\beta}\rangle \equiv\frac{1}{\sqrt{Z_{\beta}}}\sum_n e^{-\frac{\beta E_n}{2}}\vert n,n\rangle\;,
\label{TFDdef}
\ee
in the tensor product of the original Hilbert space with itself. Unitary evolution with a single Hamiltonian gives
\be 
\vert\psi_{\beta} (t)\rangle=e^{-iHt}\vert\psi_{\beta}\rangle=\vert\psi_{\beta+2it}\rangle\;.
\ee
This state has many applications. In the AdS/CFT correspondence, it is dual to the eternal black hole \cite{Maldacena_2003}. The crucial feature exploited in Ref.~\cite{SpreadC} is that the survival amplitude for the time evolved TFD state, i.e. the amplitude that time evolution leaves it unchanged, is
\be 
S(t) =\langle\psi_{\beta+2it}\vert\psi_{\beta}\rangle=\frac{Z_{\beta-it}}{Z_{\beta}}\;.
\label{eq:TFDsurvival}
\ee
The spectral form factor (SFF) is then the survival {\it probability} (i.e. the magnitude squared of (\ref{eq:TFDsurvival}) associated with the evolution of the TFD, a point that has been observed on a number of occasions (see, e.g. \cite{Papadodimas:2015xma,delCampo:2017bzr,Verlinde:2021jwu,Stanford:2022fdt}).  The authors of \cite{SpreadC} showed that this quantity controls the wavefunction of the time-evolved TFD in the Krylov basis at all times;  indeed, the entire wavefunction can be computed just from knowledge of the survival amplitude as a function of time.

Numerical methods can be used to compute the spectral form factor and spread complexity exactly for a given member of a matrix ensemble (see \cite{SpreadC}).  Repeated draws from the ensemble can then be used to compute the ensemble average. But we can also compute these quantities by using the one and two-point functions of the Lanczos coefficients to approximate the Hamiltonian of the theory. Even numerically, one advantage is that instead of requiring $O(N^2)$ matrix elements the approximation only requires $\mathcal{O}(N)$ matrix elements since it is tridiagonal. Once we have the approximation $H_{\textrm{trid}}^{\textrm{approx}}$ of $H_{\textrm{trid}}$\footnote{ By $H_{\textrm{trid}}$ we mean the exact Hamiltonian in the Krylov basis, i.e in the basis in which is tridiagonal.} we compute the SFF via
\be
    Z_\beta = \textrm{Tr}(e^{-\beta H})=\textrm{Tr}(e^{-\beta H_{\textrm{trid}}})\approx \textrm{Tr}(e^{-\beta H_{\textrm{trid}}^{\textrm{approx}}}),
\ee
and the spread complexity using the methods described in \cite{SpreadC}.

Equivalently, for a given RMT, we define a new theory of tridiagonal Hamiltonians defined by fixing the average and covariance of the diagonal and off-diagonal matrix elements. These are given by equations \eqref{eq:one_point_final} and \eqref{eq:M_expanded} respectively. To sample a Gaussian distribution with a known covariance matrix $M^{-1}$, we use the Cholesky decomposition to write the covariance as a product $M^{-1}=LL^T$. We then sample i.i.d gaussian variables $\xi_i$ with covariance $\overline{\xi_i \xi_j}=\delta_{ij}$; a linear transformation of these variables $x_i=L_{ij}\xi_j$  has covariance 
\be
\overline{x_ix_j}=L_{ik}L_{jl}\overline{\xi_k \xi_l}=L_{ik}L_{jk}=M^{-1}_{ij}\;.
\ee
Our approximate $M^{-1}$  calculated using our methods above is not precisely symmetric as translation invariance is not perfectly preserved for finite $N$. We can make it symmetric taking the average of $M^{-1}$ with its transpose. The precise symmetrization method does not affect the results at large $N$.

In Fig.~\ref{fig:reverse_check_1} we numerically compute the SFF for the three example classes of RMT studied above. In all cases, we compare results obtained from the exact distribution of eigenvalues of the matrix model and the tridiagonal version approximated by fixing the leading behavior of the one and two-point functions at large $N$. Remarkably, there is an extremely good match between the computations at all times.   Fig.~\ref{fig:reverse_check_2} shows a similar excellent match between the exact and approximate calculations of the spread complexity \eqref{eq:spread_comp_def}.

The same methods using an approximate tridiagonal Hamiltonian with the correct one- and two-point statistics can be used to compute the survival probability and spread complexity of the initial state $\ket{\psi}=(1,0,\cdots ,0)$ associated with the Lanczos coefficients computed in this paper and the corresponding Krylov basis.  However, we should expect this computation to deviate from the exact answer, because, as discussed above, our analytical formulae describe the bulk of the large-N matrix Hamiltonian, but not the ``edge'', i.e. $a_i$ and $b_i$ for $i$ of  $O(1)$.
This inaccuracy is washed out in the analysis of the Thermo-Field Double state that we described above. This is because we simply used the tridiagonal approximation as a method of generating a Hamiltonian, and then used the methods of \cite{SpreadC} to produce the SFF and spread complexity associated with the TFD.   This computation re-tridiagonalizes the Hamiltonian for the TFD state.  The associated Lanczos coefficients are then global linear combinations of the coefficients for which we derived an analytical formula.  As such, at large $N$, the ``edge'' of our Hamiltonian makes a negligible contribution to the time development of the TFD state, and the results of this section simply serve as a check on the accuracy of our results for the bulk of the Hamiltonian.  In a forthcoming paper, we will explain how to ``pad'' the edge of our analytical approximation to accurately describe the full-time evolution of the initial state $\ket{\psi}=(1,0,\cdots ,0)$, which, as mentioned above, can be considered as an initial random state.

\section{Discussion}

The Lanczos tridiagonalization method is an important tool for studying quantum theories. In the context of quantum evolution, it facilitates efficient computation of the time-evolved wavefunction. In the context of RMTs, it enables a more efficient generation of random matrix ensembles. The Lanczos algorithm also suggests a definition of the complexity of time evolution,  for both operators \cite{Parker:2018yvk} and states \cite{SpreadC}, that can be practically analyzed in a wide range of theories.

Here, we provided an analytical treatment of the statistics of the Lanczos coefficients, i.e., the components of the triadiagonalized Hamiltonian, for general Random Matrix Theories.
Mathematically speaking, our results generalize the seminal article of Dumitriu and Edelman \cite{Dumitriu_2002} concerning Gaussian beta ensembles to RMTs controlled by generic potentials, or, equivalently, having arbitrary densities of states. Specifically, we gave integral formulae for the one and two-point functions of the components of tridiagonalized Hamiltonians drawn from ensembles specified by arbitrary potentials in the large-N limit.  For polynomial potentials, we showed that the expected values and covariances of these Lanczos coefficients are determined by the solutions to certain algebraic equations. All our findings were verified numerically in a  variety of examples.

Our results are derived under the assumption that the Lanczos coefficients at large $N$ are continuous. While this holds for a wide variety of random matrix ensembles, particularly those that have a unimodal density of states, this assumption fails when the density of states is sharply multimodal such as when the spectrum is gapped. In such a case, \eqref{intdl} does not have a solution that satisfies the monotonicity assumption of section \ref{solve_intdl}, as the integrand spans an interval while the density of states is broken or is too thin in the middle of that interval to support the integrand at all $x$. This does not mean, however, that these cases are intractable; some numerical experimentation shows that rather than being approximately constant across small intervals, as we have assumed, the Lanczos coefficients in these cases become approximately \emph{periodic} across small intervals. A generalization of our methods to these periodic cases would be interesting and would tackle a wider range of problems.

We close with some remarks and future directions. Wigner envisioned random matrices as a means for understanding aspects of the spectra of heavy nuclei.  However, another prominent application of RMT concerns quantum chaos. In particular, a basic conjecture states that the statistics of random matrices approximates the fine-grained structure of the spectrum of a quantum chaotic Hamiltonian  \cite{osti_4801180,PhysRevLett.52.1} (see the reviews \cite{Guhr:1997ve,akemann2011oxford}). In this context, our new tools may assist in studying universal aspects of the wave-functions of many-body quantum chaotic systems at long-time scales. These tools may also allow the study of aspects of the quantum complexity of these systems  \cite{SpreadC}. 

For example, the authors of \cite{SpreadC} showed that the widely-studied spectral form factor of an RMT is in fact just one component of the occupancy distribution in the Krylov basis associated with the time evolution of the  Thermofield Double State \cite{Takahashi:1996zn}.\footnote{The TFD is the canonical purification of the thermal ensemble in quantum mechanical systems \cite{Takahashi:1996zn}. In the context of quantum gravity and the holographic correspondence, it describes eternal black holes \cite{Israel:1976ur,Maldacena_2003}.}  The time development of the spectral form factor is then governed by a Schrodinger equation (\ref{SchrodingerEq}) that encodes the long-time dynamics of the complete wave function in the Krylov basis. For RMTs, this Schrodinger equation is effectively a random one dimensional chain, where the hopping parameters are the Lanczos coefficients whose statistics we computed above. The study of this random Schrodinger equation may shed new light on quantum chaos. We will expand on this idea in a companion article \cite{usfuture}.

It would also be interesting to connect our ideas and framework to the recently considered effective field theory of quantum chaos, see \cite{Altland:2020ccq,Altland:2021rqn,Altland:2022xqx} and references therein, which also seeks to describe aspects of the long time dynamics of quantum chaotic systems, and the associated spectral structure. On a different note, as mentioned in \cite{SpreadC}, our approach may be a promising avenue for understanding the late time dynamics of black hole interiors, potentially connecting with \cite{Iliesiu:2021ari,Stanford:2022fdt}. Finally, it would be desirable to connect our ``physics'' approach to quantum state complexity to other approaches for studying complexity of time evolution in many body systems that have appeared recently \cite{doi:10.1126/science.1121541,Magan:2016ehs,Jefferson:2017sdb,Chapman:2017rqy,Caputa:2017urj,Magan:2018nmu,Caputa:2018kdj,Balasubramanian:2019wgd,Balasubramanian:2021mxo,Bueno:2019ajd,Erdmenger:2020sup,Chagnet:2021uvi,PRXQuantum.2.030316,Haferkamp:2021uxo,Koch:2021tvp,Czech:2022wtt,Rabambi:2022jwu}.




\section*{Acknowledgements}

We are grateful to Pawel Caputa for many discussions about spread complexity and the Lanczos approach, to Julian Sonner for conversations about quantum chaos, to Nicholas Witte for sharing his work on the thermodynamic limit of the Lanczos method, and to Roman Geiko for useful correspondence on properties of Toeplitz matrices. This work is supported by a DOE through DE-SC0013528,  QuantISED grant DE-SC0020360, and the Simons Foundation It From Qubit collaboration (385592).  VB thanks the Galileo Galileo Institute where he was a Simons Visiting Scientist, and the Aspen Center for Physics (supported by NSF grant PHY-1607611), for hospitality as this work was completed.

\appendix

\section{Recursion relations and generating functions}\label{appendix_math}

In this appendix, we prove some properties of the coefficients
\be
    C_{\alpha,\beta}^n \equiv \sum_{k=0}^{n} \binom{k}{(k+\alpha)/2}\binom{n-k}{(n-k+\beta)/2}\;,
\ee
that appeared in the Gaussian kernels derived in the main text. Our first objective is to prove the following relations for $\alpha, \beta\geq 0$
\bea
    & C^n_{\alpha, \beta}&=C^n_{0,\alpha+\beta}\equiv C^n_{\alpha+\beta}\;,\\
    & C^n_{\delta}&=C^{n-1}_{\delta-1}+C^{n-1}_{\delta+1}\;,\\
    & C^{2n}_{0}&= 4^n\;.
\eea
Note that flipping the sign of either $\alpha$ or $\beta$ in $C^n_{\alpha,\beta}$ does not change its value, due to the symmetry of the binomial coefficients.

To this end, we start by proving some properties of the generating functions of the binomial coefficients.

\begin{lem}\label{lem:lemma1}
    Let
    \be
        g_{\alpha}(x) = \sum_{n=0}^{\infty} \binom{n}{(n+\alpha)/2}\,x^n,
    \ee
    then we have the recursion relation
    \be
        x\,(g_{\alpha}(x)+g_{\alpha+2}(x))=g_{\alpha+1}(x)\;,
    \ee
    for $\alpha$ a non-negative integer.
\end{lem}

\begin{proof}
    This follows from Pascal's Triangle relation
    \be 
        \binom{n}{k}=\binom{n-1}{k}+\binom{n-1}{k-1}\;.
    \ee
\end{proof}

Notice that the first two elements of the sequence are
\be
    g_0(x)=\frac{1}{\sqrt{1-4x^2}}\;,\,\,\,\,\,\,\,\,\,\,\,\,\,\,\, g_1(x)=\frac{1-\sqrt{1-4x^2}}{2x\sqrt{1-4x^2}}\;,
\ee
which can be obtained from the well-known generating function of the Catalan numbers.

\begin{lem}\label{lem:lemma2}
    Let
    \be
        g_{\alpha}(x) = \sum_{n=0}^{\infty} \binom{n}{(n+\alpha)/2}\, x^n,
    \ee
    where the binomial coefficients are assumed to be zero if the arguments are not integers. Then the following relation
    \be
        g_{\alpha}(x) \, g_{\beta}(x)=g_{\gamma}(x) \, g_{\delta}(x)\;,
    \ee
   holds whenever $\alpha, \beta, \gamma, \delta$ are non-negative integers satisfying $\alpha+\beta=\gamma+\delta$.
\end{lem}

\begin{proof}
    We prove this by strong induction on $\max\,(\alpha, \beta, \gamma, \delta)$. Assume without loss of generality that $\max\,(\alpha, \beta, \gamma, \delta)=\alpha$.
    
    \emph{Base cases:} Consider $\alpha <3$ satisfying $\alpha+\beta=\gamma+\delta$. The only nontrivial case to show is
    \bea
        g_2(x)\,g_0(x)&=\gr{\frac{1-\sqrt{1-4x^2}}{2x^2\sqrt{1-4x^2}}}-\frac{1}{\sqrt{1-4x^2}}\frac{1}{\sqrt{1-4x^2}}\\
        &=\gr{\frac{1-2x^2-\sqrt{1-4x^2}}{2x^2\sqrt{1-4x^2}}}\frac{1}{\sqrt{1-4x^2}}\\
        &=\frac{2-4x^2-2\sqrt{1-4x^2}}{4x^2(1-4x^2)}=g_1^2(x).
    \eea
    
    \emph{Step:} Now we consider $\alpha\geq 3$.
    
    If $\alpha$ is equal to $\gamma$ or $\delta$, then $\beta$ is equal to $\delta$ or $\gamma$ respectively, and the relation holds trivially. If $\alpha$ is equal to $\beta$, then $\alpha=\beta=\gamma=\delta$ and the relation holds.
    If $\alpha=\beta+1$, then necessarily $\alpha=\gamma,\beta=\delta$ or $\alpha=\delta,\beta=\gamma$, and again the relation holds. The only remaining cases satisfy $\alpha>\beta+1, \gamma, \delta$.
    
    In this case, lemma \ref{lem:lemma1} and the fact that $\alpha\geq 3$ give us
    \be
        g_{\alpha}(x)\, g_{\beta}(x)=g_{\beta}(x)\,\gr{\frac{g_{\alpha-1}(x)}{x}-g_{\alpha-2}(x)}.
    \ee
    Using the inductive hypothesis, since $\alpha>\beta+1$, we have $g_{\alpha-1}(x) \, g_{\beta}(x)=g_{\alpha-2}(x) \, g_{\beta+1}(x)$, and $g_{\alpha-2}(x) \, g_{\beta}(x)=g_{\alpha-3}(x) \, g_{\beta+1}(x)$. We can apply lemma \ref{lem:lemma1} again to get
    \be
        =g_{\beta+1}(x)\,\gr{\frac{g_{\alpha-2}(x)}{x}-g_{\alpha-3}(x)}=g_{\alpha-1}(x) \, g_{\beta+1}(x).
    \ee
    Finally, since $\alpha>\beta+1, \gamma, \delta$, we have that $\max(\alpha-1,\beta+1, \gamma, \delta)<\alpha$, so using the inductive hypothesis again
    \be
        g_{\alpha}(x)\,g_{\beta}(x)=g_{\alpha-1}(x)\,g_{\beta+1}(x)=g_{\gamma}(x)\,g_{\delta}(x)\;.
    \ee
\end{proof}

Now we come back to the sequences that  appear directly in the Gaussian kernels, namely
\eqm{C_{\alpha, \beta}^n = \sum_{k=0}^{n} \binom{k}{(k+\alpha)/2}\binom{n-k}{(n-k+\beta)/2}
}
These are convolutions between the previous binomial sequences corresponding to $g_{\alpha},g_{\beta}$. Therefore the generating function of the sequence is
\be
\sum_{n=0}^\infty C_{\alpha, \beta}^n \, x^n
=g_{\alpha}(x)\, g_{\beta}(x)\;.
\ee
By lemma \ref{lem:lemma2}, this depends only on $\alpha+\beta$ rather than $\alpha$ and $\beta$ independently. We then define $C^n_\delta=C^n_{\alpha,\beta}$ for any $\alpha+\beta=\delta$. The next lemma establishes the properties of $C_m^n$ we wanted to prove.

\begin{lem}\label{lem:lemma3}
    $C^n_m$ satisfies 
    \bea
        C^n_0&=&\begin{cases}2^n & (n~\text{even})\\ 0 &(n~\text{odd})\end{cases}\nonumber\\
        C^n_1&=&\begin{cases}0 & (n~\text{even})\\ 2^n-\frac{1}{2}\binom{n+1}{(n+1)/2} &(n~\text{odd})\end{cases}\nonumber\\
        C^n_\delta+C^n_{\delta+2}&=&C^{n+1}_{\delta+1}\label{pascalconv_recursionrelations}
    \eea
\end{lem}

\begin{proof}
    We use the generating function $f_\delta(x)=\sum C^n_\delta\, x^n = g_0(x) \, g_\delta(x)$. For $\delta=0$, $f_0(x)=g_0(x)^2=\frac{1}{1-4x^2}$, whose Taylor expansion is given by the first relation. For $\delta=1$, $f_1(x)=g_0(x) \, g_1(x)=\frac{1}{1-4x^2}\frac{1-\sqrt{1-4x^2}}{2x}=\frac{1}{2x(1-4x^2)}-\frac{1}{2x\sqrt{1-4x^2}}$. The first term gives the exponential and the second term gives the binomial in the second relation.
    
    To show the third relation, we use lemma \ref{lem:lemma1} in the following equation
    \eqm{x\,(f_{\delta}(x)+f_{\delta+2}(x))=g_0(x)\,x\,(g_{\delta}(x)+g_{\delta+2}(x))=g_0(x) \, g_{\delta+1}(x)=f_{\delta+1}(x).
    }
    Extracting the coefficient of $x^{n+1}$ in the above equation gives the third relation.
\end{proof}

Finally, to have an intuitive idea of the binomial sequence $C_{\alpha,\beta}^n$, we depict the first few values in Fig.~(\ref{fig:table_cmn}).

\begin{figure}[t]
\begin{tabular}{c|cccccccc}
\hspace*{1cm}& \hspace*{5mm}0\hspace*{5mm} & \hspace*{5mm}1\hspace*{5mm} & \hspace*{5mm}2\hspace*{5mm}  & \hspace*{5mm}3\hspace*{5mm}  & \hspace*{5mm}4\hspace*{5mm} & \hspace*{5mm}5\hspace*{5mm} & \hspace*{5mm}6\hspace*{5mm} & \hspace*{5mm}7\hspace*{5mm} \\\hline
0 & 1                     &                         &    &    &   &   &   &   \\
1   &                       & 1                       &    &    &   &   &   &   \\
2   & 4                     &                         & 1  &    &   &   &   &   \\
3   &                       & 5                       &    & 1  &   &   &   &   \\
4   & 16                    &                         & 6  &    & 1 &   &   &   \\
5   &                       & 22                      &    & 7  &   & 1 &   &   \\
6   & 64                    &                         & 29 &    & 8 &   & 1 &   \\
7   &  & 93 &    & 37 &   & 9 &   & 1 \\
\end{tabular}
\caption{Values of $C^n_m$ for various $m$ (columns) and $n$ (rows). Blank spots are zero. The first three columns appear in the online encyclopedia of integer sequences as OEIS A000302, A000346, A008549, see \cite{oeis}.}
\label{fig:table_cmn}
\end{figure}

Our next step is to find distributions $\eta_m(x)$ whose moments correspond to the given sequences, namely
\be 
\int dx\,x^n \,  \eta_m(x)=C^n_m\;.\label{eq:rho_definition}
\ee
It is possible to find $\eta_m$ analytically by first converting the above (ordinary) generating functions into exponential generating functions via a Borel transform, and then finding $\eta_m$ via a Fourier transform of such a Borel transform along the imaginary axis. But it is easier to analyze this directly with a recursion relation.

First, notice that
\be
\eta_0(x)=\frac{1}{2}\delta(x+2)+\frac{1}{2}\delta(x-2)\implies \int dx\, x^n \,  \eta_0(x)=C^n_0\;,
\ee
where $C^n_m$ was defined in \eqref{pascalconv_recursionrelations}.

In order to find $\eta_1(x)$, we remind that $\int_{-2}^2 x^n \frac{1}{\pi \sqrt{4-x^2}}dx=\binom{n}{n/2}$. Then we have
\be
\eta_1(x)=\frac{1}{2}\delta(x-2)-\frac{1}{2}\delta(x+2)-\frac{1}{2}\frac{x}{\pi \sqrt{4-x^2}}H(4-x^2)\implies \int dx\, x^n \eta_1(x)=C^n_1
\ee
where $H(x)$ is the Heaviside step function. The recursion relation in \eqref{pascalconv_recursionrelations} is satisfied if $\eta_m(x)+\eta_{m+2}(x)=x\,\eta_{m+1}$, as follows from
\be
    C^{n+1}_{m+1}=\int dx~ x^{n+1}\, \eta_{m+1}(x)=\int dx~ x^{n}\, (\eta_m(x)+\eta_{m+2}(x))=C^n_m+C^n_{m+2}\;.
\ee
To solve this linear recurrence relation, we can write
\be
    \eta_m(x) = \frac{1}{2}\delta(x-2)+(-1)^m\frac{1}{2}\delta(x+2)+\frac{1}{2}\frac{x P_m(x)}{\pi \sqrt{4-x^2}}H(4-x^2)\;, \label{eq:rhom_definition}
\ee
where $P_m$ satisfies $P_m = xP_{m-1}-P_{m-2}$. Using the ansatz $P_{m}=r^m$, the characteristic equation is $r^2-xr+1=0$, which gives the solutions $r_\pm = e^{\pm i\theta}$, with $\theta=\cos^{-1}(\frac{x}{2})$. Solving for $P_0=0$, $P_1=1$, this gives us 
\be
    P_m(x)=Ae^{-im\theta}+Be^{im\theta}=-\frac{\sin\gr{m\cos^{-1}\gr{\frac{x}{2}}}}{\sin\gr{\cos^{-1}\gr{\frac{x}{2}}}}\;. \label{eq:Pm_definition}
\ee
Lastly, due to the recurrence relations, we know that
\be
    \eta_{\abs{m-1}}(x) \, t^m+\eta_{\abs{m+1}}(x) \, t^m=x \, \eta_{\abs{m}}(x) \, t^m\;,
\ee
except at $m=0$, where $2\,\eta_{1}(x)=x\,\eta_{0}(x)-\frac{x}{\pi\sqrt{4-x^2}}$. Summing the two-sided generating function 
\be
    \eta(x,t)\equiv \sum_{m=-\infty}^\infty \eta_{\abs{m}}(x) \, t^m\;,
\ee
we find
\bea
    \gr{t+\frac{1}{t}} \, \eta(x,t)&=&x \, \eta(x,t)-\frac{x}{\pi\sqrt{4-x^2}}H(4-x^2) \nonumber \\
    \implies \eta(x,t)&=&-\frac{1}{\gr{t+\frac{1}{t}}-x}\frac{x}{\pi\sqrt{4-x^2}}H(4-x^2)\;.\label{eq:rho_sum}
\eea

\bibliographystyle{utphys}
\bibliography{main}

\end{document}